\documentclass[onefignum,onetabnum]{siamart171218}
\pdfoutput=1

\usepackage{fullpage}
\usepackage{hyperref}
\usepackage[protrusion=true,expansion=true]{microtype}
\usepackage[utf8]{inputenc}
\usepackage[T1]{fontenc}
\usepackage{amsfonts}
\usepackage{graphicx}
\usepackage{xspace}							%
\usepackage{paralist}
\usepackage[linesnumbered, vlined,algo2e]{algorithm2e}
\usepackage{tikz}
\usetikzlibrary{calc,positioning}

\usepackage{tabularx}
\newcommand{\Trule}{\rule{0pt}{3ex}}
\newcommand{\Brule}{\rule[-1.5ex]{0pt}{0pt}}

\newcommand{\optprobA}[3]{
\begin{center}
\begin{tabularx}{\textwidth}{|l X|}
	\hline
	\multicolumn{2}{|c|}{#1\Trule} \\
	\textbf{Input:\ }&{#2}\\
	\textbf{Solution:\ }&{#3\Brule}\\
	\hline
\end{tabularx}
\end{center}
}
\newcommand{\domSet}{\textsc{Dominating Set}\xspace}
\newcommand{\setCover}{\textsc{Set Cover}\xspace}
\newcommand{\unwDirST}{\textsc{Unweighted Directed Steiner Tree}\xspace}
\newcommand{\dirST}{\textsc{Directed Steiner Tree}\xspace}
\newcommand{\unwST}{\textsc{Unweighted Steiner Tree}\xspace}
\newcommand{\ST}{\textsc{Steiner Tree}\xspace}
\newcommand{\SF}{\textsc{Steiner Forest}\xspace}
\newcommand{\PCST}{\textsc{Prize Collecting Steiner Tree}\xspace}

\newcommand{\FPT}{\ensuremath{\mathsf{FPT}}\xspace}

\newcommand{\EPAS}{\ensuremath{\mathsf{EPAS}}\xspace}
\newcommand{\PSAKS}{\ensuremath{\mathsf{PSAKS}}\xspace}

\newcommand{\W}[1]{\ensuremath{\mathsf{W[#1]}}\xspace}
\newcommand{\NP}{\ensuremath{\mathsf{NP}}\xspace}
\newcommand{\coNPp}{\ensuremath{\mathsf{coNP/Poly}}}

\newcommand{\ZTIME}{\ensuremath{\mathsf{ZTIME}}\xspace}
\newcommand{\Px}{\ensuremath{\mathsf{P}}\xspace}
\newcommand{\APX}{\ensuremath{\mathsf{APX}}\xspace}

\newcommand{\eps}{\ensuremath{\varepsilon}}

\newcommand{\bigO}[1]{\ensuremath{\operatorname{O}\left(#1\right)}}

\newcommand{\polyn}{\ensuremath{\cdot n^{\bigO{1}}}}

\newtheorem{innerReductionRule}{Reduction Rule}
\newenvironment{reductionRule}[1]
  {\renewcommand\theinnerReductionRule{#1}\innerReductionRule}
  {\endinnerReductionRule}

\newcommand{\ExtNei}[1]{N^{#1}_{\textrm{Ext}}}

\newcommand{\hy}{\hbox{-}\nobreak\hskip0pt}

\usepackage[sort,comma,square,numbers]{natbib}

\usepackage{xcolor}
\definecolor{darkblue}{rgb}{0,0,0.45}
\definecolor{darkred}{rgb}{0.6,0,0}
\definecolor{darkgreen}{rgb}{0.13,0.5,0}
\hypersetup{colorlinks, linkcolor=darkblue, citecolor=darkgreen,
urlcolor=darkblue}

\usepackage[nameinlink]{cleveref}
\usepackage{aliascnt}

\theoremstyle{plain}
\newtheorem{thm}{Theorem}%
\newaliascnt{obs}{thm}
\newaliascnt{lem}{thm}
\newaliascnt{crl}{thm}
\newaliascnt{dfn}{thm}
\newaliascnt{claim}{thm}
\newaliascnt{prop}{thm}
\newaliascnt{hyp}{thm}
\newtheorem{obs}[obs]{Observation}
\newtheorem{lem}[lem]{Lemma}

\newtheorem{crl}[crl]{Corollary}

\theoremstyle{definition}
\newtheorem{dfn}[dfn]{Definition}
\aliascntresetthe{obs}
\aliascntresetthe{lem}
\aliascntresetthe{crl}
\aliascntresetthe{prop}
\aliascntresetthe{dfn}
\aliascntresetthe{claim}
\aliascntresetthe{hyp}

\crefname{thm}{Theorem}{Theorems}
\crefname{lem}{Lemma}{Lemmas}

\usepackage{todonotes}

\newcommand{\R}{\mathbb{R}}

\title{Parameterized Approximation Schemes for Steiner Trees with Small Number of Steiner Vertices
\thanks{Submitted to the editors \today.
The research was done while the authors were at Charles University.
An extended abstract of this manuscript appeared at STACS 2018~\cite{STACS}.
\funding{This work was partially supported by the project SVV--2017--260452.
T.~Masařík was supported by project 17-09142S of GAČR.
A.E. Feldmann was supported by the Czech Science Foundation GA{\v C}R (grant \#17-10090Y), and by the Center for Foundations of Modern Computer Science (Charles Univ.\ project UNCE/SCI/004).
A.E.~Feldmann, T.~Masařík, T.~Toufar, and P.~Veselý were supported by the project GAUK 1514217.
D. Knop is supported by the OP VVV MEYS funded project CZ.02.1.01/0.0/0.0/16\_019/0000765 ``Research Center for Informatics''.
P.~Veselý was supported by European Research Council grant ERC-2014-CoG 647557.
The research leading to these results has received funding from the European Research Council under the European Union’s Seventh Framework Programme (FP/2007-2013) / ERC Grant Agreement n. 616787.
}}}

\author{
Pavel Dvo{\v{r}}{\'{a}}k\thanks{Computer Science Institute, Charles University, Prague
  (\email{koblich@iuuk.mff.cuni.cz}, \email{toufar@iuuk.mff.cuni.cz}).
  }
\and
Andreas E. Feldmann\thanks{Department of Applied Mathematics, Charles University, Prague, Czech Republic
  (\email{feldmann.a.e@gmail.com}).}
\and
Du{\v{s}}an Knop\thanks{Faculty of Information Technology, Czech Technical University, Prague, Czech Republic
  (\email{dusan.knop@gmail.com}).}
\and
Tom{\'{a}}{\v{s}} Masa{\v{r}}{\'{i}}k\thanks{Faculty of Mathematics, Informatics and Mechanics of University of Warsaw, Poland \& Department of Applied Mathematics, Charles University, Prague, Czech Republic (\email{masarik@kam.mff.cuni.cz}).}
\and
Tom{\'{a}}{\v{s}} Toufar\footnotemark[2]
\and
Pavel Vesel{\'{y}}\thanks{Department of Computer Science, University of Warwick, Coventry, UK
  (\email{pavel.vesely@warwick.ac.uk}).}
}

\usepackage{amsopn}

\ifpdf
\hypersetup{
  pdftitle={Parameterized Approximation Schemes for Steiner Trees with Small Number of Steiner Vertices},
  pdfauthor={P. Dvořák, A. E. Feldmann, D. Knop, T. Masařík, T. Toufar, and P. Veselý}
}
\fi

\begin{document}

\maketitle

\begin{abstract}
We study the~\ST problem, in which a set of \emph{terminal} vertices needs to be
connected in the cheapest possible way in an edge-weighted graph. This problem
has been extensively studied from the viewpoint of approximation and also
parameterization. In particular, on one hand \ST is known to be \APX-hard, and
\W{2}-hard on the other, if parameterized by the~number of non-terminals
(\emph{Steiner vertices}) in the optimum solution. In contrast to this, we give
an \emph{efficient parameterized approximation scheme} (\EPAS), which
circumvents both hardness results. Moreover, our methods imply the existence of
a \emph{polynomial size approximate kernelization scheme} (\PSAKS) for
the~considered parameter.

We further study the parameterized approximability of other variants of \ST,
such as \dirST and \SF. For neither of these an \EPAS is likely to exist for
the studied parameter: For \SF an easy observation shows that the problem is
\APX-hard, even if the input graph contains no Steiner vertices. For \dirST we
prove that approximating within any function of the studied parameter is
\W{1}-hard. Nevertheless, we show that an \EPAS exists for \unwDirST, but a
\PSAKS does not. We also prove that there is an \EPAS and a \PSAKS for \SF if in
addition to the~number of Steiner vertices, the~number of connected components
of an optimal solution is considered to be a~parameter.
\end{abstract}



\sloppy

\section{Introduction}
In this paper we study several variants of the \ST problem. In its most basic
form this optimization problem takes an undirected graph $G=(V,E)$ with edge
weights $w(e)\in\R^+_0$ for every $e\in E$, and a set $R\subseteq V$ of
\emph{terminals} as input. The non-terminals in $V\setminus R$ are called
\emph{Steiner vertices}. A \emph{Steiner tree} is a tree in the graph $G$, which
spans all terminals in $R$ and may contain some of the Steiner vertices. The
objective is to minimize the total weight $\sum_{e\in E(T)} w(e)$ of the
computed Steiner tree $T\subseteq G$. This fundamental optimization problem is
one of the 21 original \NP-hard problems listed by Karp \cite{MR51:14644} in his
seminal paper from 1972, and has been intensively studied since then. The \ST
problem and its variants have applications in network design, circuit layouts,
and phylogenetic tree reconstruction, among others (see
survey~\cite{hwang1992steiner}).

Two popular ways to handle the seeming intractability of \NP-hard problems are
to design \emph{approximation}~\cite{williamson2011design} and
\emph{parameterized}~\cite{pc-book} algorithms. For the former, an
\emph{$\alpha$\hy{}approximation} is computed in polynomial time for some
factor $\alpha$ specific to the algorithm, i.e., the solution is always at most
a multiplicative factor of $\alpha$ worse than the optimum of the input
instance. The \ST problem, even in its basic form as defined above, is
\APX-hard~\cite{chlebik2002approximation}, i.e., it is \NP-hard to obtain an
approximation factor of $\alpha=\frac{96}{95}\approx 1.01$. However, a factor of
$\alpha=\ln(4)+\eps\approx 1.39$ can be achieved in polynomial
time~\cite{DBLP:journals/jacm/ByrkaGRS13}, which is the currently best factor
known for this runtime.

For parameterized algorithms, an instance is given together with a
\emph{parameter} $p$ describing some property of the input.
The idea is to isolate the exponential runtime of an \NP-hard problem to the
parameter. That is,
the optimum solution is computed in time $f(p)\polyn$, where $f$ is a
computable
function independent of the input size~$n$. If such an algorithm exists, we call
the problem \emph{fixed-parameter tractable} (\FPT) for parameter $p$.
Here, the choice of the parameter is crucial, and a problem
may be \FPT for some parameters, but not for others.
A well-studied parameter for the \ST problem is the number of terminals $|R|$.
It is known since the classical result of
\citet{DBLP:journals/networks/DreyfusW71} that the \ST problem is \FPT for this
parameter, as the problem can be solved in time~$3^{|R|}\polyn$ if~$n=|V|$. A
more recent algorithm by \citet{fuchs2007dynamic} obtains runtime
$(2+\delta)^{|R|}\cdot n^{O_\delta(1)}$ for any constant $\delta>0$. This
can be improved to $2^{|R|}\polyn$ if the input graph is unweighted via the 
algorithm of \citet{DBLP:conf/icalp/Nederlof09} (using results of 
\mbox{\citet{BjorklundHKK07}}). A somewhat complementary and
less-studied parameter to the number of terminals is the number of Steiner
vertices in the optimum solution, i.e., $p=|V(T)\setminus R|$ if $T$ is an
optimum Steiner tree. It is known~\cite{downey2012parameterized} that \ST is
\mbox{\W{2}-hard} for parameter $p$ and therefore is unlikely to be \FPT, in
contrast to the parameter $|R|$. This parameter $p$ has been mainly studied in
the context of unweighted graphs before. The problem remains \W{2}-hard in this
special case and therefore the focus has been on designing parameterized
algorithms for restricted graph classes, such as planar or $d$-degenerate
graphs~\cite{jones2013parameterized,suchy2017extending}.

In contrast to this, our question is: What can be done in the most general case,
in which the class of input graphs is unrestricted and edges may have weights?
Our main result is that we can overcome the \APX-hardness of \ST on one hand,
and on the other hand also the \W{2}-hardness for our parameter of choice $p$,
by combining the two paradigms of approximation and
parameterization.\footnote{This area has recently received growing interest
(cf.~the \href{https://sites.google.com/site/aefeldmann/workshop}{Parameterized
Approximation Algorithms Workshop})}
We show that there is an \emph{efficient parameterized approximation scheme}
(\EPAS), which for any $\eps>0$ computes a $(1+\eps)$-approximation in time
$f(p,\eps)\polyn$ for a function~$f$ independent of~$n$. Note that
here we consider the approximation factor of the algorithm as a parameter as
well, which accounts for the ``efficiency'' of the approximation scheme
(analogous to an \emph{efficient polynomial time approximation scheme}
or~$\mathsf{EPTAS}$). In fact, as summarized in the following theorem, our
algorithm computes an approximation to the cheapest tree having at most $p$
Steiner vertices, even if better solutions with more Steiner vertices exist.

\begin{thm}\label{thm:ST}
There is an algorithm for \ST, which given an edge-weighted undirected graph
$G=(V,E)$, terminal set $R\subseteq V$, $\eps>0$, and an integer $p$, computes a
$(1+\eps)$\hy{}approx\-imation to the cheapest Steiner tree $T\subseteq G$ with
$p\geq|V(T)\setminus R|$ in time $2^{\bigO{p^2/\eps^4}}\polyn$. \footnote{If the
input to this optimization problem is malformed (e.g., if $p$ is smaller than
the number of Steiner vertices of any feasible solution) then the output of the
algorithm can be arbitrary (cf.~\cite{lokshtanov2017lossy})}
\end{thm}

It is worth noting that here we treat the actual value of~$p$ as a parameter; not as a ``hard constraint''.
That is, the solution returned by our algorithm may contain more than~$p$ Steiner vertices and only its quality (cost) is compared to the cost of the cheapest solution that contains at most~$p$ Steiner vertices.
This is true for all our approximation algorithms.

Many variants of the \ST problem exist, and we explore the applicability of
our techniques to some common ones. For the \dirST problem the aim is to compute
an \emph{arborescence}, i.e., a directed graph obtained by orienting the edges
of a tree so that exactly one vertex, called the \emph{root}, has in-degree zero
(which means that all vertices are reachable from the root). More concretely,
the input consists of a directed graph $G=(V,A)$ with arc weights
$w(a)\in\R^+_0$ for every~$a\in A$, a terminal set~$R\subseteq V$, and a
specified terminal $r\in R$. A \emph{Steiner arborescence} is an arborescence in
$G$ with root $r$ containing all terminals $R$. The objective is to find a
Steiner arborescence $T\subseteq G$ minimizing the weight $\sum_{a\in A(T)}
w(a)$. This problem is notoriously hard to approximate: No
$\bigO{\log^{2-\eps}(n)}$\hy{}approximation exists unless
$\NP\subseteq\ZTIME(n^{\mathrm{polylog}(n)})$~\cite{approx-hardness}. But even
for the \unwDirST problem in which each arc has unit weight, a fairly simple
reduction from the \textsc{Set Cover} problem implies that no
$((1-\eps)\ln n)$-approximation algorithm is possible unless
$\Px=\NP$~\cite{approx-hardness,dinur2014setcover}. At the same time, even
\unwDirST is \W{2}-hard for our considered parameter
$p$~\cite{molle2008enumerate,jones2013parameterized}, just as for the undirected
case. For this reason, all previous results have focused on restricted inputs:
\citet{jones2013parameterized} prove that when combining the parameter $p$ with
the size of the largest excluded topological minor of the input graph, \unwDirST
is \FPT. They also show that if the input graph is acyclic and $d$-degenerate, where degeneracy is measured in the underlying undirected graph,
the problem is \FPT for the combined parameter $p$ and $d$.

Our focus again is on general unrestricted inputs. We are able to leverage our
techniques to the unweighted directed setting, and obtain an \EPAS, as
summarized in the following theorem.
Here, the cost of a Steiner arborescence is the number of contained arcs.

\begin{thm}\label{thm:unwDirST}
There is an algorithm for \unwDirST, which given an unweighted directed graph
$G=(V,A)$, terminal set $R\subseteq V$, root $r\in R$, $\eps>0$, and
integer~$p$, computes a $(1+\eps)$-approximation to the cheapest Steiner
arborescence $T\subseteq G$ with $p\geq |V(T)\setminus R|$ in time
$2^{p^2/\eps}\polyn$. \footnotemark[2] %
\end{thm}

Can our techniques be utilized for the even more general case when arcs have
weights? Interestingly, in contrast to the above theorem we can show that in
general the \dirST problem most likely does not admit such approximation
schemes, even when allowing ``non-efficient'' runtimes of the form
$f(p,\eps)\cdot n^{g(\eps)}$ for any computable functions $f$ and $g$. This
follows from the next theorem, since setting $\eps$ to any constant, the
existence of such a $(1+\eps)$-approximation algorithm would imply $\W{1} =
\FPT$.

\begin{thm}\label{thm:dirST}
For any computable function $f$, it is \W{1}-hard to compute an
$f(p)$\hy{}approximation of the optimum Steiner arborescence $T$ for \dirST
parameterized by $p=|V(T)\setminus R|$, if the input graph is arc-weighted.
\end{thm}

Another variant of \ST is the \textsc{Node Weighted Steiner Tree} problem, in
which the Steiner vertices have weights, instead of the edges. The aim is to
minimize the total weight of the Steiner vertices in the computed solution. A
similar reduction as the one used to prove \cref{thm:dirST} (from
\textsc{Dominating Set}) shows that also in this case computing any
$f(p)$-approximation is \W{1}-hard, even if all Steiner vertices have unit
weight.

Other common variants of \ST include the \PCST and \SF problems. The latter 
takes as input an edge-weighted undirected graph $G=(V,E)$ and a list 
$\{s_1,s'_1\},\ldots,\{s_k,s'_k\}$ of terminal pairs, i.e., the set of terminals 
is $R=\{s_i,s'_i\mid 1\leq i\leq k\}$. A \emph{Steiner forest} is a forest $F$ 
in $G$ for which each $\{s_i,s'_i\}$ pair is in the same connected component, 
and the objective is to minimize the total weight of the forest $F$. For this 
variant it is not hard to see that parameterizing by $p=|V(F)\setminus R|$ 
cannot yield any approximation scheme, as a simple reduction from \ST shows that 
the problem is \APX-hard even if the input has no Steiner vertices (see 
\cref{sec:ST-to-SF}). For the \PCST problem, the input is again a terminal set 
in an edge-weighted graph, but the terminals have additional costs. A solution 
tree is allowed to leave out a terminal but has to pay its cost in return 
(cf.~\cite{williamson2011design}). It is also not hard to see that this problem 
is \APX-hard, even if there are no Steiner vertices at all.

These simple results show that our techniques to obtain approximation schemes
reach their limit quite soon: With the exception of \unwDirST, most common
variants of \ST seem not to admit approximation schemes for our parameter $p$.
We are however able to generalize our \EPAS to \SF if we combine $p$ with the
number $c$ of connected components in the optimum solution. In fact, our main
result of \cref{thm:ST} is a corollary of the next theorem, using only the first
part of the above mentioned reduction from \ST (cf.~\cref{sec:ST-to-SF}). Due to
this, it is not possible to have a parameterized approximation scheme for the
parameter $c$ alone, as such an algorithm would imply a polynomial time
approximation scheme for the \APX-hard \ST problem. Hence the following result
necessarily needs to combine the parameters $p$ and $c$.

\begin{thm}\label{thm:SF}
There is an algorithm for \SF, which given an edge-weighted undirected graph 
$G=(V,E)$, a list $\{s_1,s'_1\},\ldots,\{s_k,s'_k\}\subseteq V$ of terminal 
pairs, $\eps>0$, and integers~$p,c$, computes a $(1+\eps)$-approximation to the 
cheapest Steiner forest $F\subseteq G$ with at most $c$ connected components and 
$p\geq |V(F)\setminus R|$ where $R=\{s_i,s'_i\mid 1\leq i\leq k\}$, in 
time~$(2c)^{\bigO{(p+c)^2/\eps^4}}\polyn$. \footnotemark[2]
\end{thm}

As mentioned for \cref{thm:ST}, our algorithm might compute an approximate 
solution with more than $p$ Steiner vertices. Analogously, it may also compute 
a forest with more than $c$ components, even if its cost is compared to the 
best one containing at most $p$ Steiner vertices and $c$ components only.

A topic tightly connected to parameterized algorithms is kernelization. We here
use the framework of \citet{lokshtanov2017lossy}, who also give a thorough
introduction to the topic (see~\cref{sec:kernel-prelims} for formal
definitions). Loosely speaking, a \emph{kernelization algorithm} runs in
polynomial time, and, given an instance of a parameterized problem, computes
another instance of the same problem, such that the size of the latter instance
is at most $f(p)$ for some computable function $f$ in the parameter $p$ of the
input instance. The computed instance is called the \emph{kernel}, and for an
optimization problem it must be possible to efficiently convert an optimum
solution to the kernel into an optimum solution to the input instance.

A fundamental result of parameterized complexity says that a problem is \FPT if
and only if it has a kernelization algorithm~\cite{pc-book}. This means that for
our parameter $p$, most likely \ST does not have a kernelization algorithm, as
it is \W{2}-hard. For this reason, the focus of kernelization results have
previously shifted to special cases again. By a folklore result, \ST is \FPT for
our parameter $p$ if the input graph is planar
(cf.~\cite{jones2013parameterized}). Of particular interest are \emph{polynomial
kernels}, which have size polynomial in the input parameter. The idea is that
computing the kernel in this case is an efficient preprocessing procedure for
the problem, such that exhaustive search algorithms can be used on the kernel.
\citet{suchy2017extending} proved that \unwST parameterized by $p$ admits a
polynomial kernel if the input graph is planar.

Our aspirations again are to obtain results for inputs that are as general as
possible, i.e., on unrestricted edge-weighted input graphs. We prove that \ST
has a polynomial \emph{lossy} (approximate) kernel, despite the fact that the problem is
\W{2}-hard: An \emph{$\alpha$-approximate kernelization algorithm} is a
kernelization algorithm that computes a new instance for which a given
$\beta$-approximation can be converted into an $\alpha\beta$-approximation for
the input instance in polynomial time. The new instance is now called a
\emph{(polynomial) approximate kernel}, and its size is again bounded as a
function (a polynomial) of the parameter of the input instance.

Just as for our parameterized approximation schemes in \cref{thm:ST,thm:SF}, we
prove the existence of a lossy kernel for \ST by a generalization to \SF where
we combine the parameter $p$ with the number $c$ of connected components in the
optimum solution. Also, our lossy kernel can approximate the optimum arbitrarily
well: We prove that for our parameter the \SF problem admits a \emph{polynomial
size approximate kernelization scheme} (\PSAKS), i.e., for every $\eps>0$ there
is a $(1+\eps)$-approximate kernelization algorithm that computes a polynomial
approximate kernel. An easy corollary then is that \ST parameterized only by~$p$
also has a \PSAKS, by setting $c=1$ in \cref{thm:PSAKS-SF} and using the above
mentioned reduction from \ST to \SF (cf.~\cref{sec:ST-to-SF}).

\begin{thm}\label{thm:PSAKS-SF}
There is a $(1+\eps)$-approximate kernelization algorithm for \SF, which given
an edge-weighted undirected graph $G=(V,E)$, a list of terminal pairs
\mbox{$\{s_1,s'_1\},\ldots,\{s_k,s'_k\}\subseteq V$}, and
integers~$p,c$, computes an approximate kernel of size
$\left((p+c)/\eps\right)^{2^{\bigO{1/\eps}}}$, if for the optimum Steiner
forest $F\subseteq G$, we have $p\geq |V(F)\setminus R|$ where 
$R=\{s_i,s'_i\mid 1\leq i\leq k\}$, the number of connected components of $F$ is 
at most $c$, and $\eps>0$. \footnotemark[2]
\end{thm}

Analogous to approximation schemes, it is possible to distinguish between
efficient and non-efficient kernelization schemes: A \PSAKS is \emph{size efficient}
if the size of the approximate kernel is bounded by $f(\eps)\cdot p^{\bigO{1}}$,
where $p$ is the parameter and $f$ is a computable function independent of~$p$.
Our bound on the approximate kernel size in \cref{thm:PSAKS-SF}
implies that we do \emph{not} obtain a size efficient \PSAKS for either \SF or \ST.
This is in contrast to the existence of efficient approximation schemes for the
same parameters in \cref{thm:ST,thm:SF}. We leave open whether or not a size efficient
\PSAKS can be found in either case. Interestingly, we also do not obtain any
\PSAKS for the \unwDirST problem, even though by \cref{thm:unwDirST} an \EPAS
exists. In fact, we prove the following theorem.

\begin{thm}\label{thm:no-PSAKS-unwDirST}
No $(2-\eps)$-approximate kernelization algorithm exists for the \unwDirST problem
parameterized by the number $p=|V(T)\setminus R|$ of Steiner vertices in the
optimum Steiner arborescence $T$ for any $\eps>0$, unless $\NP\subseteq\coNPp$.
\end{thm}

\subsection{Used techniques}
Our algorithms are based on the intuition that a Steiner tree containing only
few Steiner vertices but many terminals must either contain a large component
induced by terminals, or a Steiner vertex with many terminal neighbors
forming a large star. A~high\hy{}level description of our algorithms for
\unwDirST and \SF therefore is as follows. In each step a tree is found in the
graph in polynomial time, which connects some terminals using few Steiner
vertices. We save this tree as part of the approximate solution and then
contract it in the graph. The vertex resulting from the contraction is declared
a terminal and the process repeats for the new graph. Previous
results~\cite{jones2013parameterized,suchy2017extending} have also built on this
straightforward procedure in order to obtain \FPT algorithms and polynomial
kernels for special cases of \unwDirST and \unwST. In particular, in the
unweighted undirected setting it is a well-known fact
(cf.~\cite{suchy2017extending}) that contracting an adjacent pair of terminals
is always a safe option, as there always exists an optimum Steiner tree
containing this edge. However, this immediately breaks down if the input graph is
edge-weighted, as an edge between terminals might be very costly and should
therefore not be contained in any (approximate) solution.

Instead, we employ more subtle contraction rules, which use the following
intuition. Every time we contract a tree with $\ell$ terminals we decrease the
number of terminals by $\ell-1$ (as the vertex arising from a contraction is a
terminal). Our ultimate goal would be to reduce the number of terminals to
one---at this point, the edges that we contracted during the whole run connect
all the terminals. Decreasing the number of terminals by one can therefore be
seen as a ``unit of work''. We will pick a tree with the lowest cost per unit of
work done, and prove that as long as there are sufficiently many terminals left
in the graph, these contractions only lose an $\eps$-factor compared to the
optimum. As soon as the number of terminals falls below a~certain threshold
depending on the given parameter, we can use an \FPT algorithm computing the
optimum solution in the remaining graph. This algorithm is parameterized by the
number of terminals, which now is bounded by our parameter. For the variants of
\ST considered in our positive results, such \FPT algorithms can easily be
obtained from the ones for
\ST~\cite{DBLP:journals/networks/DreyfusW71,BjorklundHKK07,fuchs2007dynamic}.
Adding this exact solution to the previously contracted trees gives a feasible
solution that is a $(1+\eps)$-approximation.

Each step in which a tree is contracted in the graph can be seen as a
\emph{reduction rule} as used for kernelization algorithms. Typically, a proof
for a kernelization algorithm will define a set of reduction rules and then show
that the instance resulting from applying the rules exhaustively has size
bounded as a function in the parameter. To obtain an $\alpha$-approximate
kernelization algorithm, additionally it is shown that each reduction rule is
\emph{$\alpha$-safe}. Roughly speaking, this means that at most a factor of
$\alpha$ is lost when applying any number of $\alpha$-safe reduction rules (see
\cref{sec:kernel-prelims} for formal definitions).

Contracting edges in a directed graph may introduce new paths, which did not
exist before. Therefore, for the \unwDirST problem, we need to carefully choose
the arborescence to contract. In order to prove \cref{thm:unwDirST}, we show that
each contraction is a $(1+\eps)$-safe reduction rule. However, the total size
of the graph resulting from exhaustively applying the contractions is not
necessarily bounded as a function of our parameter. Thus, we do not obtain an
approximate kernel.

For \SF the situation is in a sense the opposite. Choosing a tree to contract
follows a fairly simple rule. On the downside however, the contractions we
perform are not necessarily $(1+\eps)$-safe reduction rules. In fact there are
examples in which a single contraction will lose a large factor compared to the
optimum cost. We are still able to show that after performing all
contractions exhaustively, any $\beta$-approximation to the resulting instance
can be converted into a $(1+\eps)\beta$-approximation to the original input
instance. Even though the total size of the resulting instance again cannot be
bounded in terms of our parameter, for \SF we can go on to obtain a \PSAKS. For
this we utilize a result of~\citet{lokshtanov2017lossy}, which shows how to
obtain a \PSAKS for \ST if the parameter is the number of terminals. This result
can be extended to \SF, and since our instance has a number of terminals bounded
in our parameter after applying all contractions, we obtain \cref{thm:PSAKS-SF}.

To obtain our inapproximability result of \cref{thm:dirST}, we use a reduction
from the \domSet problem. It was recently shown by \citet{dom-set} that this
problem does not admit parameterized $f(k)$-approximation algorithms for any
function $f$, if the parameter $k$ is the solution size, unless $\W{1} = \FPT$.
We are able to exploit this to also show that no such algorithm exists for
\dirST with edge weights, under the same assumption. To prove
\cref{thm:no-PSAKS-unwDirST} we use a cross composition from the \setCover
problem, for which \citet{dinur2014setcover} proved that it is \NP-hard to
compute a $(1-\eps)\ln(n)$-approximation.
We are able to preserve only a constant gap; thus, we leave open whether stronger non-constant lower
bounds are obtainable, or whether \unwDirST has a polynomial size
$\alpha$-approximate kernel for some constant $\alpha\geq 2$.

\subsection{Related work}\label{sec:related}
As the \ST problem and its variants have been studied since decades, the
literature on this topic is huge. We only present a selection of related work
here, that was not yet mentioned above.

For general input graphs, \citet{zelikovsky_1993_11_over_6_apx} gave the first
polynomial time approximation algorithm for \ST with a better ratio than $2$
(which can easily be obtained by computing an MST on the terminal set). His
algorithm is based on repeatedly contracting stars with three terminals each, in
the metric closure of the graph, which yields a $11/6$-approximation. This line
of work led to the \citet{borchers-du} Theorem, which states that for every \ST
instance with terminal set $R$ and every $\eps>0$ there exists a set of
sub-trees (so-called \emph{full components}) on at most $2^{O(1/\eps)}$
terminals from $R$ each and with all leaves being terminals, such that their union forms a Steiner tree for $R$ of
cost at most $1+\eps$ times the optimum. As a consequence, it is possible to
compute all full components with at most $2^{O(1/\eps)}$ terminals (using an
\FPT algorithm parameterized by the number of
terminals~\cite{DBLP:journals/networks/DreyfusW71,fuchs2007dynamic}), and then
find a subset of the precomputed solutions, in order to approximate the optimum.
This method is the basis of most modern \ST approximation algorithms, and is for
instance leveraged in the currently best $(\ln(4)+\eps)$-approximation algorithm
of \citet{DBLP:journals/jacm/ByrkaGRS13}. The \mbox{\citet{borchers-du}}
Theorem can also be interpreted in terms of approximate kernelization schemes,
as \citet{lokshtanov2017lossy} point out (cf.~proof of \cref{thm:PSAKS-SF}). It
is interesting to note that our algorithms are also based on finding good
sub-trees. However, while computing optimum full components is NP-hard, the
sub-trees we compute in each step can be found in polynomial time, regardless of
how many terminals they contain.

For planar graphs~\cite{DBLP:journals/talg/BorradaileKM09} it was shown that an
$\mathsf{EPTAS}$ exists for \ST. For \SF a $2$-approximation can be computed in
polynomial time on general inputs~\cite{DBLP:journals/siamcomp/AgrawalKR95}, but
an $\mathsf{EPTAS}$ also exists if the input is
planar~\cite{eisenstat2012efficient}. If the \unwST problem is parameterized by
the solution size, it is known~\cite{dom} that no polynomial (exact) kernel
exists, unless $\mathsf{NP}\subseteq\mathsf{coNP}/\mathsf{Poly}$. If the input is
restricted to planar or bounded-genus graphs it was shown that polynomial
kernels do exist for this parameterization~\cite{planar-7}. It was later
shown~\cite{suchy2017extending} that for planar graphs this is even true for our
smaller parameter $p$.

For the \dirST problem it is a long standing open problem whether a
polylogarithmic approximation can be computed in polynomial time. It is known
that an \mbox{$\bigO{|R|^\eps}$-approximation} can be computed in polynomial
time~\cite{DBLP:journals/jal/CharikarCCDGGL99}, and an $\bigO{\log^2
n}$-approximation in quasi-polynomial
time~\cite{DBLP:journals/jal/CharikarCCDGGL99}.
\citet{DBLP:conf/icalp/FeldmannM16} consider the \textsc{Directed Steiner
Network} problem, which is the directed variant of \SF (i.e.,~a generalization of
\dirST). They give a dichotomy result, proving that the problem parameterized
by $|R|$ is \FPT whenever the terminal pairs induce a graph that is a
caterpillar with a constant number of additional edges, and otherwise the
problem is \W{1}-hard. Among the \W{1}-hard cases is the \textsc{Strongly
Connected Steiner Subgraph} problem (for which the hardness was originally
established by \citet{DBLP:journals/siamdm/GuoNS11}), in which all terminals
need to be strongly connected. For this problem a $2$-approximation is
obtainable~\cite{DBLP:conf/iwpec/ChitnisHK13} when parameterizing by~$|R|$, and a
recent result shows that this is the best possible~\cite{chitnis2017parameterized}
under the Gap Exponential Time Hypothesis.

In the same paper, \citet{chitnis2017parameterized} also consider the
\textsc{Bidirected Steiner Network} problem, which is the directed variant of
\SF on \emph{bidirected} input graphs, i.e., directed graphs in which for every
edge $uv$ the reverse edge $vu$ exists as well and has the same cost. These
graphs model inputs that lie between the undirected and directed settings. From
\cref{thm:ST,thm:PSAKS-SF}, it is not hard to see that the \textsc{Bidirected
Steiner Tree} problem (i.e.,~\dirST on bidirected inputs) has both an \EPAS and
a \PSAKS for our parameter $p$, by reducing the problem to the undirected
setting. Since the \PSAKS for parameter $p$ follows from the \PSAKS for
parameter~$|R|$ given by \citet{lokshtanov2017lossy}, it is interesting to note
that for parameter~$|R|$, \citet{chitnis2017parameterized} provide both a
\PSAKS and a parameterized approximation scheme for the \textsc{Bidirected
Steiner Network} problem whenever the optimum solution is planar. This is
achieved by generalizing the \mbox{\citet{borchers-du}} Theorem to this setting. 
As this is a generalization of \textsc{Bidirected Steiner Tree}, it is natural 
to ask whether corresponding algorithms also exist for our parameter $p$ in the 
more general setting considered in~\cite{chitnis2017parameterized}.

\section{Preliminaries}
\subsection{Reducing Steiner tree to Steiner forest}
\label{sec:ST-to-SF}

By a folklore result, we may reduce the \ST problem to \SF. For this we pick an
arbitrary terminal~$r$ of the \ST instance, and for every other terminal~$v$ of
this instance, introduce a terminal pair $\{v,r\}$ for \SF.

If we want to construct an instance without Steiner vertices, we can add a new
vertex $w'$ for every Steiner vertex $w$ of \ST and add an edge $ww'$ of cost
$0$. Additionally, we introduce a terminal pair $\{w,w'\}$ to our \SF instance.
Hence, $R=V$ in the constructed \SF instance (i.e., there are no Steiner 
vertices), but an optimum Steiner forest in the constructed
graph costs exactly as much as an optimum Steiner tree in the original graph. As
\ST is \APX-hard, the same is true for \SF, even if all vertices are terminals.

\subsection{Lossy kernels}\label{sec:kernel-prelims}

We give a brief introduction to the lossy kernel framework as introduced
by~\citet{lokshtanov2017lossy}. See the latter reference for a thorough
introduction to the topic.

For an optimization problem, a \emph{polynomial time preprocessing algorithm}
is a pair of polynomial time algorithms: the \emph{reduction algorithm}
$\mathcal{R}$ and the \emph{solution lifting algorithm} $\mathcal{L}$. The
former takes an instance~$I$ with parameter $p$ of a given problem as input, and
outputs another instance $I'$ with parameter~$p'$. The solution lifting
algorithm $\mathcal{L}$ converts a solution for the instance $I'$ to a solution
of the input instance $I$: Given a solution $s'$ to $I'$, $\mathcal{L}$ computes
a solution $s$ for $I$ such that $s$ is optimal for $I$ if $s'$ is optimal for
$I'$. If additionally the output of $\mathcal{R}$ is bounded as a function of
$p$, i.e., when $|I'|+p'\leq f(p)$ for some computable function $f$ independent
of~$|I|$, then the pair given by $\mathcal{R}$ and $\mathcal{L}$ is called a
\emph{kernelization algorithm}, and $I'$ together with parameter $p'$ is the
\emph{kernel}. If the reduction and solution lifting algorithms get an input
that is not an instance of the problem (for example if the parameter does not
correctly describe some property of the optimum solution), then the outputs of
the algorithms are undefined and can be arbitrary.

An \emph{$\alpha$-approximate polynomial time preprocessing algorithm} is again
a pair of a reduction algorithm~$\mathcal{R}$ and a solution lifting algorithm
$\mathcal{L}$, both running in time polynomial in the input size. The reduction
and solution lifting algorithms are as before, but there is a different property
on the output of the latter: If the given solution $s'$ to the instance $I'$
computed by $\mathcal{R}$ is a $\beta$-approximation, then the output of
$\mathcal{L}$ is a solution $s$ that is an $\alpha\beta$-approximation for the
original instance~$I$. Analogous to before, an \emph{$\alpha$-approximate
kernelization algorithm} is an $\alpha$-approximate polynomial time
preprocessing algorithm for which the size of the output of the reduction
algorithm is bounded in terms of $p$ only. The output of $\mathcal{R}$ is in
this case called an \emph{approximate kernel}, and it is \emph{polynomial} if
its size is bounded by a polynomial in $p$.

In the context of lossy kernels, a \emph{reduction rule} is a reduction algorithm
$\mathcal{R}$. It is called \emph{$\alpha$\hy{}safe} if a solution lifting
algorithm $\mathcal{L}$ exists, which together with $\mathcal{R}$ form a
\emph{strict} $\alpha$-approximate polynomial time preprocessing algorithm.
This means that if $s'$ is a $\beta$-approximation for the instance computed by
$\mathcal{R}$, then $\mathcal{L}$ computes a
$(\max\{\alpha;\beta\})$-approximation $s$ for the input instance. As shown
in~\cite{lokshtanov2017lossy}, the advantage of considering this stricter
definition is that, as usual, reduction rules can be applied exhaustively, until
a stable point is reached in which none of the rules would change the instance
any longer. The algorithm resulting from applying these rules, together with
their corresponding solution lifting algorithms, forms a strict
$\alpha$-approximate polynomial time preprocessing algorithm (which is not
necessarily the case when using the non-strict definition;
see~\cite{lokshtanov2017lossy}).

\section{The weighted undirected Steiner forest and Steiner tree problems}
\label{sec:newalgo}

In this section we describe an approximate polynomial time preprocessing
algorithm that returns an instance of \SF containing at most
$\bigO{(p+c)^2/\eps^4}$ terminals if there is a Steiner forest with at most $p$ Steiner
vertices and at most $c$ connected components.
We can use this algorithm in two ways.
Either we can proceed with a kernelization derived from
\citet{lokshtanov2017lossy} and obtain a polynomial size lossy kernel
(\cref{thm:PSAKS-SF}), or we can run an exact \FPT algorithm derived
from \citet{fuchs2007dynamic} on the reduced instance,
obtaining an \EPAS running in single exponential time with respect to the
parameters (\cref{thm:SF,thm:ST}).
In both cases we use the combined parameter $(p,c)$.

\optprobA
{\SF}
{A graph $G = (V,E)$, with edge weights $w(e)\in \R^+$ for each $e\in E$, and a
list $\{s_1,s'_1\},\ldots,\{s_k, s'_k\}$ of pairs of terminals.}
{A Steiner forest $F \subseteq G$ containing an $s_i$-$s'_i$ path for
every $i \in [k]$}

\subsection{Algorithm description}
We first rescale all weights so that every edge has weight strictly greater than 
$1$. Using a standard preprocessing procedure, we also take the metric closure 
of the input graph, i.e., every edge of the graph is present and its weight is 
equal to the shortest path distance of the endpoints in the original input 
graph. It is easy to see (and folklore) that solving \SF in the metric closure 
is equivalent to solving it for the original input graph. Moreover, every 
solution still exists as a subgraph in the metric closure, so that our 
parameters remain unchanged.

Then, in each step of our algorithm we pick a star, add it to the solution, and 
contract the star in the current graph. After the contraction, the edge weights 
may not obey the triangle inequality anymore. However, this is not needed for 
our algorithm. Instead, we only need that the graph is always complete, so that 
a star to contract can always be found.
We repeat this procedure until the number of terminals falls below a specified 
bound depending on $\eps$, $p$, and~$c$. To describe how we pick the star to be 
contracted in each step, we need to introduce the \emph{ratio} of a star. 

\begin{dfn}
Let $C$ be a set of edges of a star, i.e., all edges of $C$ are incident to a 
common vertex which is the \emph{center} of the star, and denote by $Q$ the set 
of terminals in~$V(C)$, where~$V(C)$ is the set of vertices in~$C$. Provided 
$|Q| \geq 2$, we define the \emph{ratio} of $C$ as $w(C) / (|Q|-1)$, where 
$w(C)=\sum_{e\in C} w(e)$. 
\end{dfn}

Note that we allow $C$ to contain only a single edge if it joins two terminals, 
and that due to rescaling of edge weights each star has ratio strictly greater 
than $1$. Observe also that the ratio of a star is similar to the average weight
of an edge in the star. However the ratio is skewed due to the subtraction of~$1$ in the 
denominator. In particular, for two stars of the same average weight, the one 
with more terminals will have the smaller ratio. Thus, in this sense, picking a 
star with small ratio favors large stars.

In every step, our algorithm contracts a star with the best available ratio
(i.e., the lowest ratio among all stars connecting at least two terminals). 
Since we assume that our input is a complete graph, a star containing two 
terminals always exists (except in the trivial case when there is only one 
terminal). Moreover, due to the following lemma, a star with the best ratio has 
a simple form: It consists of the cheapest $\ell$ edges incident to its center 
vertex such that all leaves are terminals. As there are $n$ possible center 
vertices and at most $n$ incident edges to each center which can be sorted in 
time $\bigO{n \log n}$, the best ratio star can be found in time $\bigO{n^2 \log 
n}$.

\begin{lem}\label{lem:min-star}
Let $v$ be a vertex and denote by $q_1,q_2,\ldots$ the terminals adjacent
to~$v$, where $w(vq_1)\leq w(vq_2)\leq\cdots$, i.e., the terminals are ordered
non-decreasingly by the weight of the corresponding edge~$vq_i$. The star with
the best ratio having $v$ as its center has edge set $\{vq_1, vq_2, \ldots,
vq_\ell\}$ for some $\ell$.
\end{lem}

\begin{proof}
Let $C$ be an edge set of a star with center vertex $v$. First note that if
this star contains a Steiner vertex $w$ as a leaf, $vw$ can be removed from $C$
in order to decrease the ratio $w(C)/(|Q|-1)$, since only the terminals $Q$ of
the star are counted in the denominator. Also if $C$ does not contain some edge
$vq_i$ but an edge $vq_j$ with~$j>i$, then we may switch the edge $vq_j$ for
$vq_i$ in $C$ in order to optimize the ratio: The denominator stays the same,
but the numerator cannot increase, as the terminals $q_1,q_2,\ldots$ are
ordered non-decreasingly according to the weights of $vq_i$.
\end{proof}

We now formally describe different graphs resulting from each contraction step 
$t$, together with their terminal pairs. Initially, we set $G_0$ to the input 
graph, and in each step $t\geq 0$ we obtain a new graph $G_{t+1}$ from $G_t$ by 
contracting a set of edges $C_t$ in $G_t$, such that $C_t$ forms a star of 
minimum ratio in $G_t$. That is, we obtain~$G_{t+1}$ from $G_t$ by identifying 
all vertices in~$V(C_t)$, removing all resulting loops, and among the resulting 
parallel edges we delete all but the lightest one with respect to their weights. 
We also adjust the terminal pairs in a straightforward way: Let $v$ be the 
vertex of $G_{t+1}$ resulting from contracting~$C_t$. If $G_t$ had a terminal 
pair $\{s,s'\}$ such that $s$ is incident to some edge of $C_t$ while $s'$ is 
not (i.e., $s \in V(C_t)$ and $s' \notin V(C_t)$), then we introduce the 
terminal pair $\{v,s'\}$ for $G_{t+1}$. Also every terminal pair $\{s,s'\}$ of 
$G_t$ for which neither $s$ nor~$s'$ is incident to any edge of $C_t$ is 
introduced as a terminal pair of~$G_{t+1}$. Somewhat counter-intuitively, we 
also introduce the (trivial) terminal pair $\{v,v\}$ for $G_{t+1}$ if there was 
a terminal pair in $G_t$ for which both $s$ and $s'$ were incident to edges of 
$C_t$. In particular, this means that $v$ can be a leaf of a contracted star
in a subsequent step, even though the solution might not require any 
connection from $v$ to some other terminal. The reason we need to keep $v$ as a 
terminal is that otherwise the number of Steiner vertices of the considered 
solution, i.e., our parameter $p$, might increase. Still, our analysis below 
shows that contracting such a trivial terminal $v$ in a best-ratio star will not 
cause any problems.

The number of terminals in any given instance with terminal pairs 
$\{s_1,s'_1\},\ldots,\{s_k,s'_k\}$ is the size of the set $R=\{s_i,s'_i\mid 
1\leq i\leq k\}$. This in particular means that if a terminal appears in several pairs or is in a 
trivial terminal pair, it is only counted once. The algorithm stops 
contracting best-ratio stars when there are less than~${\tau}$ terminals left in $R$;
we have ${\tau} = \bigO{(p+c)^2/\eps^4}$, but we specify the precise value of 
$\tau$ in the analysis below. If the algorithm stops in step $\tilde t$, the solution lifting 
algorithm takes a feasible solution $F$ of $G_{\tilde t}$ and returns the union 
of $F$ and $\bigcup_{t=0}^{\tilde t}C_t$. Such a solution is clearly feasible, 
since we adapted the terminal pairs accordingly after each contraction. 
\cref{alg:weightedUndirSF} gives a pseudo-code of the resulting algorithm.

\begin{algorithm}[h!]
\SetKwInOut{Input}{input}\SetKwInOut{Output}{output}
\Input{undirected graph~$G = (V,E)$, list of terminal pairs
$\{s_1,s'_1\},\ldots,\{s_k,s'_k\}$, edge weights $w(e)\in\R^+_0$}
\Output{a forest $F \subseteq G$ that contains an $s_i$-$s'_i$ 
path for any $i\in\{1,\dots,k\}$}

\SetKwFunction{BestStar}{BestStar}
\SetKwProg{Fn}{Function}{}{}
\Fn{\BestStar{v}}{
\lIf{$v$ is a terminal}{$z \gets 1$}\lElse{$z \gets 0$}

$q_1, \ldots, q_k \gets$ terminals adjacent to $v$ sorted by the weight of edge
$vq_i$\;
	\tcc{$k\ge |R| - 1$ in the metric closure}

\For{$i$ in $2-z,\dots,k$}{ %
	$r_i \gets \sum_{j=1}^i w(vq_j)/(i + z - 1)$\;
}

\Return{edges $\{vq_1,\ldots,vq_i\}$ of a star with the smallest $r_i$}
}

\While{$|R|\geq\tau$}{
	$C \gets \arg\min \{ w(C_v) \mid C_v \gets \mathtt{BestStar}(v), v\in V\}$\;
	\tcc{a star exists in the metric closure}
	
	Contract $C$, then remove loops and among parallel edges, keep only the lightest.
	Adjust the terminal pairs accordingly.
}
Run \FPT algorithm parameterized by the number of terminals and connected components\label{step:fpt_alg}\;

\caption{An algorithm for solving \SF. If we stop before line~\ref{step:fpt_alg} we
obtain the reduced instance.}\label{alg:weightedUndirSF}
\end{algorithm}

\subsection{Analysis}
For the purpose of analysis, we consider a solution in the current graph $G_t$
that originates from a solution of the original instance $G_0$, but may
contain edges that are heavier than those in $G_t$.
More concretely, denote by $F^\ast_0$ a solution in $G_0$ with
at most~$p$ Steiner vertices and at most $c$ components, i.e.,
$F^\ast_0$ is a Steiner forest containing an $s_i$-$s'_i$ path for any $i$.
We remark that $F^\ast_0$ may or may not be an optimum solution of $G_0$,
and that we think of $F^\ast_0$ as a subgraph of $G_0$, isomorphic to a forest,
without isolated vertices.

Given $F^\ast_t$ for $t\geq 0$, we modify this solution to obtain a new feasible solution
$F^\ast_{t+1}$ on the terminal pairs of~$G_{t+1}$ (as defined above).
$F^\ast_{t+1}$ will again be a subgraph of~$G_{t+1}$ without isolated vertices. Note that the edges of the
contracted star $C_t$ might not be part of $F^\ast_t$. We still mimic the
contraction of the star in $F^\ast_t$: To obtain $F^\ast_{t+1}$ from~$F^\ast_t$,
we identify all leaves of $C_t$ (which are terminals by \cref{lem:min-star} and
thus part of the solution $F^\ast_t$) and possibly also the center $v$ of $C_t$
if it is in $F^\ast_t$. (Note that if $v$ is not a terminal, it may not be a part
of the solution $F^\ast_t$, which does not contain isolated vertices.)
This results in a vertex~$v'$ and a solution $F'_{t+1}$ for $G_{t+1}$, which 
however may well contain
some cycles or loops.

We now want to delete edges incident to $v'$ in such a way that we are left with 
an acyclic feasible solution for $G_{t+1}$. Let $Q_t$ denote the set of 
terminals in~$V(C_t)$. We repeat the following simple step to find an 
inclusion-wise minimal feedback edge set $D_t$ of $F'_{t+1}$ that is incident to 
$v'$: As long as there is a cycle $K$ in $F'_{t+1}$ (possibly, $K$ is a loop), 
remove from $F'_{t+1}$ an edge $e$ of $K$ such that in solution $F^\ast_t$,
edge $e$ is incident to $Q_t$ (thus, in particular, $e$ is incident to $v'$ in 
$F'_{t+1}$). We claim that $K$ must contain an edge $e$ that is incident to a 
terminal in~$Q_t$ in solution $F^\ast_t$. Indeed, observe first that $K$ must 
contain $v'$, since otherwise $K$ appears in $F^\ast_t$, which contradicts the 
acyclicity of $F^\ast_t$. Recall that the only vertex of $V(C_t)$ that may be a 
Steiner vertex is the center $v$ of the star $C_t$. If $K$ is a loop, then the 
only edge $e$ of $K$ connects two vertices in $V(C_t)$, so $e$ is incident to 
$Q_t$. Otherwise, $K$ contains two edges $e'$ and $e''$ incident to $v'$ that do 
not connect two vertices in $V(C_t)$, because edges connecting vertices in
$V(C_t)$ become loops after the contraction. Since both $e'$ and $e''$ cannot be incident to $v$ in $F^\ast_t$
(otherwise, $K$ would be a cycle in $F^\ast_t$), one of $e'$ or $e''$ must be 
incident to $Q_t$, which shows the claim. It follows that the above procedure is 
well-defined.

Once there is no cycle in $F'_{t+1}$, we set $F^\ast_{t+1} := F'_{t+1}$,
which now forms a forest connecting all
terminal pairs of~$G_{t+1}$. Note that for each edge in $F^\ast_{t+1}$ there is
a corresponding edge in~$G_{t+1}$, which however may be lighter in~$G_{t+1}$, as
from each bundle of parallel edges in $G_t$ we keep the lightest one, but this
edge may not exist in $F^\ast_t$.
Let $D_t:=E(F^\ast_t)\setminus E(F^\ast_{t+1})$ be the set of edges that were deleted from the solution.
(We remark that we do not optimize the total length of edges in $D_t$.)

\begin{figure}[bt]
\centering
\usetikzlibrary{positioning,calc}

\begin{tikzpicture}
\tikzstyle{terminal}=[draw,  black, fill, text=white]
\tikzstyle{steiner}=[draw, black!30, circle , fill, text=black, inner sep=2pt]
\tikzstyle{edgeOpt}=[draw, thick]
\tikzstyle{edgeForbidden}=[draw, ultra thick]
\tikzstyle{edgeContracted}=[draw, dashed]

\def\vdist{40}
\def\udist{20}
\def\dist{10}
\tikzset{
    between/.style args={#1 and #2}{
         at = ($(#1)!0.5!(#2)$)
    }
}

%

\begin{scope}[local bounding box=Gt]
\node[terminal](t1){};
\node[right=\dist pt of t1](dots1){$\cdots$};
\node[right=\dist pt of dots1, terminal](t2){};
\node[below=\udist pt of dots1,steiner](u_1){};
\node[draw=none, left=\dist pt of t1](z){};
\draw[edgeOpt] (t1) -- (u_1);
\draw[edgeOpt] (t2) -- (u_1);

\node[right=\dist pt of t2, terminal](t3){};
\node[right=\dist pt of t3](dots2){$\cdots$};
\node[right=\dist pt of dots2, terminal](t4){};
\node[below=\udist pt of dots2,steiner](u_l){};
\node[right=\udist pt of u_l,terminal](t7){};
\draw[edgeOpt] (t3) -- (u_l);
\draw[edgeOpt] (t4) -- (u_l);
\draw[edgeOpt] (t7) -- (u_l);

\node[between=u_1 and u_l](centerF1){};
\node[below=10pt of centerF1,steiner](u_2){};
\draw[edgeOpt] (u_1) -- (u_2);
\draw[edgeOpt] (u_2) -- (u_l);

\node[right=20 pt of t4,terminal](t5){};

\draw[edgeOpt] (t5) -- (t4);

\node[between=t1 and t5](center){};
\node[above=\vdist pt of center,steiner](v){$v$};
\draw [edgeContracted] (t1) -- (v);
\draw [edgeContracted] (t2) -- (v);
\draw [edgeContracted] (t3) -- (v);
\draw [edgeContracted] (t4) -- (v);
\draw [edgeContracted] (t5) -- (v);
\draw [edgeForbidden] (u_1) to[out=180, in=270] (z) to[out=90, in=180] (v);

\node[left=20pt of t1](label1){};
\node[above=\vdist pt of label1](labelTop){};
\node[between=label1 and labelTop](label2){$C_t$};
\node[below=\udist pt of label1](labelBottom){$F_t^*$};
\end{scope}

\begin{scope}[local bounding box=contr, xshift=200pt, yshift=20pt]
\node[terminal](t6){$v'$};
\node[below=\udist pt of t6](belowt6){};
\node[below=\udist pt of belowt6,steiner](u_3){};
\node[left=\dist pt of belowt6,steiner](u_4){};
\node[right=\dist pt of belowt6,steiner](u_5){};
\draw[edgeOpt] (u_3) -- (u_4);
\draw[edgeOpt] (u_3) -- (u_5);

\node[left=\dist pt of t6](c1){};
\draw[edgeOpt] (u_4) to[out=90, in=180] (t6);
\draw[edgeOpt] (u_4) -- (t6);

\node[right=\dist pt of t6](c2){};
\draw[edgeOpt] (u_5) to[out=90, in=0] (t6);
\draw[edgeOpt] (u_5) -- (t6);

\node[above=\udist pt of c1](c3){};
\node[above=\udist pt of c2](c4){};
\draw[edgeOpt] (t6) .. controls (c3) and (c4) .. (t6);

\node[right=\udist pt of u_5,terminal](t8){};
\draw[edgeOpt] (u_5) -- (t8);

\draw[edgeForbidden] (u_4) to[out=135,in=150] (t6);
\end{scope}

\begin{scope}[local bounding box=Gtt, xshift=290pt, yshift=20]
\node[terminal](t9){$v'$};
\node[below=\udist pt of t9](belowt9){};
\node[below=\udist pt of belowt9,steiner](u_6){};
\node[left=\dist pt of belowt9,steiner](u_7){};
\node[right=\dist pt of belowt9,steiner](u_8){};
\draw[edgeOpt] (u_6) -- (u_7);
\draw[edgeOpt] (u_6) -- (u_8);

\node[right=\udist pt of u_8,terminal](t10){};
\draw[edgeOpt] (u_8) -- (t10);

\draw[edgeForbidden] (u_7) to[out=135,in=150] (t9);

\node[right=\dist pt of u_6] (label4) {$F^*_{t+1}$};
\end{scope}


\node[left=15pt of t6](a2){};
\node[left=of a2] (a1){};
\draw[->,thick,>=stealth] (a1) -- (a2);

\node[right=10pt of t6](a3){};
\node[right=of a3](a4){};
\draw[->,thick,>=stealth] (a3) -- (a4);

\end{tikzpicture}
  \caption{An example of creating $F^\ast_{t+1}$ from $F^\ast_t$ after contracting $C_t$.
  Solid edges (including the thick one) belong to solutions $F^\ast_{t+1}$ and $F^\ast_t$,
  while edges in $C_t$ are dashed. Note that in this example,
  no edge in $C_t$ belongs to $F^\ast_{t}$, although this is not true in general.
  Set $D_t$ consists of all edges deleted in the second step, i.e., all edges 
  incident to $v'$, except for the thick edge,
  which cannot be in $D_t$ because it is not incident to any terminal.}
  \label{fig:contraction}
\end{figure}

To show that our algorithm only loses an $\eps$-factor compared to the cost of
the solution $F^\ast_0$, we will compare the cost of the edges $C_t$
contracted by our algorithm to the set $D_t$ o deleted edges of $F^\ast_t$. 
Note that for any two time steps $t\neq t'$, the sets $D_t$ and~$D_{t'}$, but 
also the sets $C_t$ and~$C_{t'}$, are disjoint. Thus if $w(C_t) \leq 
(1+\eps)w(D_t)$ for every~$t$, then our algorithm computes a 
$(1+\eps)$\hy{}approximation. Unfortunately, this is not always the case: there 
are contractions for which this condition does not hold (see 
\cref{fig:badContraction}) and we have to account for them differently.

\begin{figure}[hbt]
\centering
\usetikzlibrary{positioning,calc}

\begin{tikzpicture}[node distance=.8cm]
\tikzstyle{terminal}=[draw, fill, black, text=white]
\tikzstyle{steiner}=[draw, black!30, circle , fill, text=black, inner sep=2pt]
\tikzstyle{edgeOpt}=[draw, ultra thick]

\def\dist{1.2}

\node[steiner] at (2,4) (s0) {};
\node[steiner] at (2,6) (s1) {};
\node[steiner] at (2,8) (s2) {};
\node[steiner] at (2,10) (s3) {};
\node[steiner] at (2,12) (s4) {};
\node[steiner] at (-2,12) (s6) {};
\node[steiner] at (-2,10) (s7) {};
\node[steiner] at (-2,8) (s8) {};
\node[steiner] at (-2,6) (s9) {};
\node[steiner] at (-2,4) (s10) {};
\node[terminal] at (-2,2) (t2) {$t_2$} ;
\node[terminal] at (2,2) (t1) {$t_1$} ;
\node[terminal] at (-4,4) (t3) {} ;
\node[terminal] at (4,4) (t4) {} ;
\node[terminal] at (-4,6) (t5) {} ;
\node[terminal] at (4,6) (t6) {} ;
\node[terminal] at (-4,8) (t7) {} ;
\node[terminal] at (4,8) (t8) {} ;
\node[terminal] at (-4,10) (t9) {} ;
\node[terminal] at (4,10) (t10) {} ;
\node[terminal] at (-4,12) (t11) {} ;
\node[terminal] at (4,12) (t12) {} ;

\node[] at (0,12) (dots) {$\cdots$};

\draw[edgeOpt] (s0) -- (t4) node[right=of s0,below] {$M$};
\draw[edgeOpt] (s1) -- (t6) node[right=of s1,below] {$M$};
\draw[edgeOpt] (s2) -- (t8) node[right=of s2,below] {$M$};
\draw[edgeOpt] (s3) -- (t10) node[right=of s3,below] {$M$};
\draw[edgeOpt] (s4) -- (t12) node[right=of s4,below] {$M$};
\draw[edgeOpt] (s6) -- (t11) node[left=of s6,below] {$M$};
\draw[edgeOpt] (s7) -- (t9) node[left=of s7,below] {$M$};
\draw[edgeOpt] (s8) -- (t7) node[left=of s8,below] {$M$};
\draw[edgeOpt] (s9) -- (t5) node[left=of s9,below] {$M$};
\draw[edgeOpt] (s10) -- (t3) node[left=of s10,below] {$M$};
\draw[edgeOpt] (s10) -- (t3) node[left=of s10,below] {$M$};

\draw[edgeOpt] (s0) -- (s1) node[midway, left=7pt] {$1$};
\draw[edgeOpt] (s2) -- (s1) node[midway, left=7pt] {$1$};
\draw[edgeOpt] (s2) -- (s3) node[midway, left=7pt] {$1$};
\draw[edgeOpt] (s3) -- (s4) node[midway, left=7pt] {$1$};
\draw[edgeOpt] (s6) -- (s7) node[midway, left=7pt] {$1$};
\draw[edgeOpt] (s7) -- (s8) node[midway, left=7pt] {$1$};
\draw[edgeOpt] (s8) -- (s9) node[midway, left=7pt] {$1$};
\draw[edgeOpt] (s10) -- (s9) node[midway, left=7pt] {$1$};
\draw[edgeOpt] (dots) -- (s6) node[midway, below] {$1$};
\draw[edgeOpt] (dots) -- (s4) node[midway, below] {$1$};

\draw (t1) -- (t2) node[midway, below] {$M$};
\draw[edgeOpt] (t1) -- (s0) node[midway, left=7pt] {$1$};
\draw[edgeOpt] (t2) -- (s10) node[midway, left=7pt] {$1$};

\end{tikzpicture}
  \caption{An example of a bad contraction. The optimum solution consists of 
the thick edges. The numbers of terminals and the weight $M$ can be arbitrarily 
large. Note that any edge in the metric closure between any two terminals has
length of at least $M$ if there are at least $M+1$ terminals.
The star centered at $t_1$ and containing the incident terminal $t_2$
has ratio $M$, while every other star in the metric closure of the graph 
has ratio slightly more than $M$. By contracting the star $t_1,t_2$ we 
create a cycle in the optimum solution containing edges of weight $1$ only. 
Thus, for a sufficiently large value of $M$ the contraction cannot be charged.}
  \label{fig:badContraction}
\end{figure}

\begin{dfn}
If $w(C_t) \leq (1+\eps)w(D_t)$ we say that the contracted edge set $C_t$ in
step~$t$ is \emph{good}; otherwise $C_t$ is \emph{bad}. Moreover, if $F^\ast_t$
has strictly more components than $F^\ast_{t+1}$, we say that $C_t$ is
\emph{multiple-component}, otherwise it is \emph{single-component}.
\end{dfn}

Our goal is to show that the total weight of bad contractions is bounded by an
$\eps$-fraction of the weight of $F^\ast_0$. 
We start by proving that if the set $Q_t$ of terminals in a contracted edge set
$C_t$ is sufficiently large, then the contraction is good. 
Intuitively, this means that skewing the ratio such that large stars are 
favored (compared to just picking the star with the smallest average weight) tends 
to result in good contractions. We define
\[
{\lambda} := \frac{(1+\eps)(p+c)}{\eps}.
\]

\begin{lem}\label{lem:largeContraction}
If $|Q_t|\ge {\lambda}$, then the contracted edge set $C_t$ is good.
\end{lem}
\begin{proof}
Let $r = w(C_t) / (|Q_t| - 1)$ be the ratio of the contracted star, and let 
$\ell'$ be the number of deleted edges in $D_t$ that connect two terminals. 
Note that any such edge has weight at least $r$, since it spans a star with two 
terminals, which has ratio equal to its weight, and since each edge in 
$F^\ast_t$ (of which $D_t$ is a subset) can only be heavier than the 
corresponding edge in the current graph $G_t$.

Let $u_1, \dots, u_q$ be the Steiner vertices adjacent to edges in $D_t$, and 
let $\ell_i$ be the number of edges in $D_t$ incident to one such Steiner 
vertex $u_i$ (see \cref{fig:bigStar}). Since $D_t$ is a feedback edge set in 
which any edge was incident to a terminal in $Q_t$ before the contraction, 
there is no edge in $D_t$ which connects two Steiner vertices. Consider the 
star spanned by the $\ell_i$ edges of $D_t$ incident to $u_i$. If $\ell_i\geq 
2$, the ratio of this star is at least $r$, since its edges are at least as 
heavy as the corresponding edges in $G_t$ and the algorithm chose a star with 
the minimum ratio in $G_t$. Thus, the weight of edges in $D_t$ incident to 
$u_i$ is at least $r(\ell_i - 1)$. In the case where $\ell_i = 1$, the lower 
bound $r(\ell_i - 1)=0$ on the weight holds trivially.

\begin{figure}[bt]
\centering
\usetikzlibrary{positioning,calc}

\begin{tikzpicture}
\tikzstyle{terminal}=[draw, fill, black]
\tikzstyle{steiner}=[draw, black!30, circle , fill, text=black, inner sep=2pt]
\tikzstyle{edgeOpt}=[draw, thick]

\def\vdist{60}
\def\udist{40}
\def\dist{10}
\tikzset{
    between/.style args={#1 and #2}{
         at = ($(#1)!0.5!(#2)$)
    }
}

\def\centerarc(#1)(#2:#3:#4)
    {  ($(#1)+({#4*cos(#2)},{#4*sin(#2)})$) arc (#2:#3:#4) }


\node[terminal](t1){};
\node[right=\dist pt of t1](dots1){$\cdots$};
\node[right=\dist pt of dots1, terminal](t2){};
\node[below=\udist pt of dots1,steiner](u_1){$u_1$};
\draw (t1) -- (u_1);
\draw (t2) -- (u_1);

\node[right=\dist pt of t2](dots3){$\cdots \cdots$};

\node[right=\dist pt of dots3, terminal](t3){};
\node[right=\dist pt of t3](dots2){$\cdots$};
\node[right=\dist pt of dots2, terminal](t4){};
\node[below=\udist pt of dots2,steiner](u_l){$u_q$};
\draw (t3) -- (u_l);
\draw (t4) -- (u_l);

\node[right=of t4, terminal](t5){};
\node[right=of t5, terminal](t6){};
\node[right=\dist pt of t6](dots4){$\cdots$};
\node[right=\dist pt of dots4, terminal](t7){};

\node[between=t5 and t7](middle_5_7){};
\node[below=of middle_5_7](middle_down_5_7){};

\draw (t5) -- (t6);
\draw (t5) .. controls (middle_down_5_7) .. (t7);
\draw(t4) -- (t5);

\node[between=t1 and t7](center){};
\node[above=\vdist pt of center,steiner](v){$v$};
\draw (t1) -- (v);
\draw (t2) -- (v);
\draw (t3) -- (v);
\draw (t4) -- (v);
\draw (t5) -- (v);
\draw (t6) -- (v);
\draw (t7) -- (v);

\draw \centerarc(u_1)(50:130:25pt);
\draw \centerarc(u_l)(50:130:25pt);

\node[right=40pt of t7](label1){$Q$};
\node[above=\vdist pt of label1](labelTop){};
\path (label1) -- (labelTop) node[midway] {$C$};
\node[below=\udist pt of label1](labelBottom){};
\path (label1) -- (labelBottom) node[midway] {$D$};
\node[below=5pt of middle_down_5_7](label2){$\ell'$};
\node[above right = 0pt and 10pt of u_1](label3){$\ell_1$};
\node[above right = 0pt and 10pt of u_l](label4){$\ell_q$};

\end{tikzpicture}
 \caption{The contracted star $C_t$ and a part of the optimal solution spanned 
by the terminals $Q_t$ of the star $C_t$.}
 \label{fig:bigStar}
\end{figure}

Any edge in $D_t$ not incident to any Steiner vertex $u_i$ connects two 
terminals. Therefore, we have $\ell' + \sum_{i=1}^q \ell_i = |D_t|$ as any edge 
in $D_t$ is incident to a terminal in $Q_t$ and we thus do not count any edge 
twice.

We observe next that from the construction of~$F^\ast_t$ we get that there are 
at least $|Q_t| - c$ edges in~$D_t$. Recall that we contracted terminals 
in~$Q_t$ in the forest $F^\ast_t$ which has at most $c$ connected components in 
order to obtain $F^\ast_{t+1}$. Indeed, a forest on $n$ vertices and $c$ 
components has $n-c$ edges. We decrease the number of vertices of $F^\ast_t$ by 
at least $|Q_t| - 1$ (one more if the center of the star with edge set $C_t$ was 
a Steiner vertex present in $F^\ast_t$), and we decrease the number of 
components by at most~$c-1$. Let $z$ be the number of vertices in $F^\ast_t$. We 
conclude that the forest $F^\ast_{t+1}$ has at most~$z - |Q_t| + 1$ vertices 
and at least~$1$ connected component. Thus, there are at most $z - |Q_t|$ edges 
in~$F^\ast_{t+1}$ and we get that $|D_t| \ge z - c - (z - |Q_t|) = |Q_t| - c$.

Since $F^\ast_t$ contains at most $p$ Steiner vertices we have $q\le p$, and we 
obtain
\[
	w(D_t) \geq r\ell' + \sum_{i=1}^q r(\ell_i - 1) = r \left( \ell' +
\sum_{i=1}^q  \ell_i  - q \right) = r(|D_t|-q)
	\geq r\bigl(|Q_t| - p - c\bigr)\,.
\]
Finally, using $|Q_t|\ge {\lambda}$ we bound $w(C_t)$ by $(1+\eps)w(D_t)$ as 
follows:
\begin{multline*}
	(1+\eps)w(D_t) \geq
	(1+\eps)r\bigl(|Q_t| - p - c\bigr) =
	r|Q_t| + r\big(\eps |Q_t| - (1+\eps)(p+c) \big)\\ \geq
	w(C_t) + r\left(\eps\frac{(1+\eps)(p+c)}{\eps}  - (1+\eps)(p+c) \right)
=
	w(C_t)\,.
\end{multline*}
\end{proof}

Note that there may be a lot of contractions with $|Q| < {\lambda}$. However, we
show
that only a bounded number of them is actually bad. The key idea is to consider
contractions with ratio in an interval $\bigl((1+\delta)^i; (1+\delta)^{i+1}\bigr]$ for
some $\delta > 0$ and integer~$i$.
Due to the rescaling of weights every star belongs to an interval with $i \geq 0$.
The following crucial lemma of our
analysis shows that the number of bad single-component contractions in each
such interval is bounded in terms of $p$ and $\varepsilon$, if $\delta$ is a
function of~$\eps$. In particular, let $\delta := \sqrt{1+\eps} - 1$, so that
$(1+\delta)^2 = 1+\eps$. We call an edge set $C$ with ratio $r$ in the $i$-th
interval, i.e., with $r\in \bigl((1+\delta)^i; (1+\delta)^{i+1}\bigr]$, an
\emph{$i$-contraction}, and define
\[
{\kappa} := \frac{(1+\delta) p}{\delta} + p\,.
\]

\begin{lem}\label{lem:UB-numOfNotChargeable}
For any integer $i$ the number of bad single-component $i$-contractions
is at most ${\kappa}$.
\end{lem}

\begin{proof}
Let us focus on bad single-component $i$\hy{}contractions only, which we here 
just call bad $i$\hy{}contractions for brevity. Suppose for a contradiction that 
the number of bad $i$-contractions is larger than ${\kappa}$. The plan is to 
show that at each of the $\kappa$ steps $t$ in which a bad $i$-contraction 
happens, there must be a cheap edge $e_t$ in the corresponding set $D_t$. Since 
the deleted sets $D_t$ are disjoint, all of these edges are also present in 
$G_{\tilde{t}}$ of the first step $\tilde{t}$ with a bad $i$-contraction, i.e., 
${\tilde{t}}$ is the minimum among all $t$ for which $w(C_t)>(1+\eps)w(D_t)$ and 
$w(C_t)/\bigl(|Q_t|-1\bigr)\in \bigl((1+\delta)^i; (1+\delta)^{i+1}\bigr]$ and 
the contraction is single-component. We then show that among all the cheap edges 
in $G_{\tilde{t}}$ there is a ``light'' star with ratio at most $(1+\delta)^i$, 
and consequently the algorithm would do a $j$-contraction for some $j<i$. This 
leads to a contradiction, since we assumed that in step ${\tilde{t}}$ the 
contraction has ratio in interval~$i$. Note that it is sufficient to find such a 
light star in $F^\ast_{\tilde{t}}$ as for each edge in $F^\ast_{\tilde{t}}$ 
there is an edge in the graph $G_{\tilde{t}}$ between the same vertices of the 
same weight or even lighter.

We claim that for each step $t$ in which the algorithm does a bad
$i$-contraction there is an edge $e_t\in D_t$ with weight at most
$(1+\delta)^{i-1}$. We have $w(C_t) > (1+\eps)w(D_t)$ as $C_t$ is bad
and $w(C_t) \le (1+\delta)^{i+1}\bigl(|Q_t| - 1\bigr)$ as the ratio of $C_t$ is in
interval $i$. Putting it together and using the definition of $\delta$ we
obtain
\[
w(D_t) < \frac{(1+\delta)^{i+1}}{1+\eps}\bigl(|Q_t| - 1\bigr) = (1+\delta)^{i-1}\bigl(|Q_t| - 1\bigr)\,.
\]
Because $C_t$ is single-component, we have $|D_t| \geq |Q_t| - 1$ and
therefore there is an edge $e_t\in D_t$ with
weight at most $(1+\delta)^{i-1}$, which proves the claim.

Note that the edge $e_t$ also exists at time step $\tilde{t}$, as $\tilde{t}\leq 
t$ and $F^\ast_{t}$ is obtained from $F^\ast_{\tilde{t}}$ by a sequence of edge 
contractions and deletions. At time $\tilde{t}$ it cannot be that $e_t$ connects 
two terminals, since we assume that the algorithm picked a star of ratio more 
than $(1+\delta)^i$ in step~$\tilde{t}$ (recall that each edge connecting two 
terminals is a star with ratio equal to its weight). It may happen though that 
$e_t$ connects two Steiner vertices in step~$\tilde{t}$. We discard any such 
edge $e_t$ that connects two Steiner vertices in step $\tilde{t}$. That is, let 
$S$ be the set of light edges~$e_t$ that connect a Steiner vertex and a terminal 
in step $\tilde{t}$. Note that edges $e_t$ and $e_{t'}$ for steps $t < t'$ with 
bad $i$-contractions are distinct, because $D_t\cap D_{t'} = \emptyset$ as all 
edges in $D_t$ are deleted from $F^\ast_t$. There are at most $p-1$ edges 
$e_t\notin S$ connecting two Steiner vertices in $F^\ast_{\tilde{t}}$, since 
$F^\ast_{\tilde{t}}$ is a forest and the solution from which 
$F^\ast_{\tilde{t}}$ is derived contained at most $p$ Steiner vertices. 
Summarizing, we assume that there are more than ${\kappa}$ bad single-component 
$i$-contractions, each of which contributes one edge $e_t$ that is incident to a 
Steiner vertex, and we remove less than $p$ edges $e_t$ that connect two 
Steiner vertices, which implies that set $S$ of the remaining edges $e_t$ 
satisfies $|S| > {\kappa} - p$.

At step ${\tilde{t}}$ there must be a Steiner vertex $v$ in $F^\ast_{\tilde{t}}$
incident to at least $|S| / p > ({\kappa} - p) / p \geq (1+\delta) / \delta$
edges in $S$. Consider a star $C$ with $v$ as the center and with edges from $S$
that are incident to~$v$; we have $|C| \ge (1+\delta) / \delta$. The ratio of
this star is at most $|C|(1+\delta)^{i-1} / \bigl(|C| - 1\bigr)$. Since $|C|  / \bigl(|C| - 1\bigr)\le (1+\delta)$ (by a routine calculation) we get that the ratio of $C$ is at
most $(1+\delta)^i$ which is a contradiction to the assumption that the
algorithm does an $i$-contraction in step~${\tilde{t}}$.
\end{proof}

We remark that the proof of \cref{lem:UB-numOfNotChargeable} does not use 
that the number of terminals in a bad $i$-contraction is bounded by ${\lambda}$, 
as shown in \cref{lem:largeContraction}. Instead we will bound the total weight 
of bad contractions in terms of ${\lambda}$. For this let $j$ be the largest 
interval of any contraction during the whole run of the algorithm, i.e., the 
ratio of every contracted star is at most $(1+\delta)^{j+1}$. As there are at 
most $\kappa$ bad single-component contractions in each interval and at most 
$c-1$ (bad) multiple-component contractions, and as the interval size grows 
exponentially, we can upper bound the total weight of bad contractions in terms 
of $\kappa, c, \lambda$, and $(1 + \delta)^j$. We can also lower bound the 
weight of $w(F^\ast_0)$ in terms of $(1 + \delta)^j$ and the lower bound $\tau$ 
on the number of terminals in the graph. If $\tau$ is large enough, then the 
total weight of edge sets $C_t$ of bad contractions is at most $\varepsilon\cdot 
w(F^\ast_0)$. These ideas are summarized in the next lemma.

\begin{lem}\label{lem:weightOfNotChargeable}
 Let $j$ be the largest interval of any contraction during the whole run of the algorithm
 and let $W_B$ be the total weight of the union edge sets~$C_t$ of bad contractions.
 Then, the following holds.
 \begin{enumerate}[(1)]
  \item $W_B \le ({\kappa} + c)  \cdot\lambda \cdot (1+\delta)^{j+2} / 
\delta\,.$
\item $w(F^\ast_0) \ge (1+\delta)^j\cdot ({\tau} - c -p)$.
\item Let
\[
{\tau} :=  ({\kappa} + c)  \cdot\lambda \cdot \frac{(1+\delta)^2}{\eps\delta}
+ c + p\,,
\]
then $W_B\le \eps\cdot w(F^\ast_0)$.
 \end{enumerate}
\end{lem}

\begin{proof}
To prove (1), observe first that if $C_t$ is a multiple-component edge set, 
$F^\ast_{t+1}$ must have at least one component fewer than $F^\ast_{t}$.
Since $F^\ast_0$ has at most~$c$ components, there are less than $c$ bad 
multiple-component contractions.
Each of them has at most $\lambda$ terminals by \cref{lem:largeContraction} and 
has ratio at most $(1+\delta)^{j+1}$ by the choice of~$j$. Thus, the total 
weight of all bad multiple-component contractions can be bounded by 
$(1+\delta)^{j+1}\cdot c \cdot \lambda$.

Note that it follows from \cref{lem:largeContraction,lem:UB-numOfNotChargeable}
that the total weight of bad single-component $i$-contractions is at most
${\kappa} \cdot {\lambda} \cdot (1+\delta)^{i+1}$.
The bound on the total weight of bad contractions follows by summing over all
intervals in which the algorithm possibly does a contraction:
\begin{align*}
{\kappa} &\cdot {\lambda} \cdot \sum_{1\le i\le j} (1 + \delta)^{i+1} + c \cdot \lambda \cdot (1+\delta)^{j+1}
\\
&= {\kappa} \cdot {\lambda} \cdot \frac{(1+\delta)^{j + 2} - (1+\delta)}{(1 + \delta) - 1} + c \cdot \lambda \cdot (1+\delta)^{j+1}
\\
&\le ({\kappa} + c)  \cdot\lambda\cdot \frac{(1+\delta)^{j + 2}}{\delta}\,.
\end{align*}
This proves (1).

For (2), when our algorithm contracted a star having ratio $r \ge (1+\delta)^j$ 
in the largest interval $j$ in some step $t$, all stars in $G_t$ with at least 
two terminals had ratio at least $r$. Let $v_1, \ldots, v_q$ be the Steiner 
vertices of $F^\ast_t$ and $u_1,\dots,u_{q'}$ be Steiner vertices of $F^\ast_t$ which are connected to at least one terminal. 
Thus, if $\ell_i$ is the number of terminals adjacent to 
$u_i$ in~$F^\ast_t$, then these terminals together with $u_i$ form a star of 
weight at least $r\cdot (\ell_i - 1)$ if $\ell_i\geq 2$, since no edge in 
$F^\ast_t$ is lighter than the corresponding edge of $G_t$. If $\ell_i=1$ then 
lower bound $r\cdot (\ell_i - 1)=0$ on the weight trivially holds. Similarly, 
all edges between terminals in $F^\ast_t$ have weight at least $r$; let $\ell'$ 
be the number of such edges.

Since $F^\ast_t$ has at least ${\tau}$ terminals in step $t$ (otherwise the
algorithm would have terminated), it contains $q$ Steiner vertices, and has at 
most $c$ components, the total number of edges of $F^\ast_t$ is $\tau+q-c$. 
Those of its edges that connect two Steiner vertices form a forest on at most 
$q$ vertices, and there can therefore be at most $q-1$ such edges. Hence the 
number of edges in $F^\ast_t$ that are incident to a terminal is $\ell' + 
\sum_{i=1}^{q'} \ell_i \ge \tau +q - c - (q - 1) \ge {\tau} - c$. Using $p\geq q'$, 
the total  weight of edges in $F^\ast_t$ is at least
\[
\ell' r + \sum_{i=1}^{q'} r\cdot (\ell_i - 1) \ge r \cdot ({\tau} - c -p)\ge
(1+\delta)^j\cdot ({\tau} - c -p)\,.
\]
This shows (2) as $w(F^\ast_t)\leq w(F^\ast_0)$.

To get (3), by (2) and using the value of ${\tau}$ we have
\[
\eps \cdot w(F^\ast_0) \ge \eps (1+\delta)^j\cdot ({\tau} - c -p) =
\eps(1+\delta)^j\cdot ({\kappa} + c)  \cdot\lambda \cdot
\frac{(1+\delta)^2}{\eps\delta} = ({\kappa} + c)  \cdot\lambda \cdot 
\frac{(1+\delta)^{j+2}}{\delta},
\]
which is the upper bound on $W_B$ by (1).
\end{proof}

The above lemma can now be used to prove that all the contractions put
together (with $\eps$ scaled appropriately) form a $(1+\eps)$-approximate preprocessing
procedure with respect to $F^\ast_0$ (cf.~\cref{sec:kernel-prelims}).

\begin{lem}\label{lem:weightedPreprocessing}
The algorithm outputs an instance with $\tau \in \bigO{(p+c)^2/\eps^4}$
terminals and (together with the solution lifting algorithm) it is a $(1 +
2\eps)$-approximate polynomial time preprocessing algorithm with respect to
$F^\ast_0$.
\end{lem}

\begin{proof}
The upper bound on the number of terminals follows directly from
the description of the algorithm. To bound the running time, we already noted
that finding a minimum ratio star to contract can be done in $\bigO{n^2 \log n}$
time. Since such a star with at least two vertices is contracted in each step
$t$ to form the next graph $G_{t+1}$, the total time used for contractions until
only ${\tau}$ terminals are left is polynomial in $n$.

Let us focus on the $(1+2\eps)$-approximate part. Let $H=G_{\tilde t}$ be the
graph left after the last contraction step $\tilde t$, and let $F_H$ be a
Steiner forest for the remaining terminal pairs. The solution lifting algorithm
simply adds all contracted edge sets $C_0,C_1,\ldots$ to $F_H$ in order to
compute a Steiner forest $F_G$ in the input graph $G_0$. We need to show that, if
$F_H$ is a $\beta$-approximation to the optimum solution $F^\ast_H$ in $H$, the resulting
forest $F_G$ is a $\bigl((1+2\eps)\beta\bigr)$-approximation to the solution
$F^\ast_0$ of $G_0$.

Let us call a step $t$ of the algorithm \emph{good} (\emph{bad}) if the
corresponding contracted edge set $C_t$ is good (bad). As all sets $C_t$ are
disjoint, we use \cref{lem:weightOfNotChargeable} to bound the weight of $F_G$ by
\[
w(F_G)= \sum_{\text{good }t} w(C_t) + \sum_{\text{bad } t} w(C_t) + w(F_H)
\leq \sum_{\text{good }t}(1+\eps) w(D_t) + \eps\cdot w(F^\ast_0) + \beta\cdot
w(F^\ast_H).
\]

The forest $F^\ast_{\tilde t}$ left after the last contraction corresponds to a
feasible solution in $H$. As the edge weights might be less expensive in $H$
than in $F^\ast_{\tilde t}$, we have $w(F^\ast_H)\leq w(F^\ast_{\tilde t})$. At
the same time, the deleted sets $D_t$ and the edges of $F^\ast_{\tilde t}$ are
disjoint, so that $\sum_{\text{good }t} w(D_t)\leq \sum_{t} w(D_t)\leq
w(F^\ast_0)-w(F^\ast_{\tilde t})$. Therefore, the above bound becomes
\begin{align*}
w(F_G)
&\leq (1+\eps) \left(w(F^\ast_0)-w(F^\ast_{\tilde t})\right)+ \eps\cdot
w(F^\ast_0) + \beta\cdot w(F^\ast_{\tilde t})
\\
&\leq (1+\eps)\beta\left(w(F^\ast_0)-w(F^\ast_{\tilde t})+w(F^\ast_{\tilde
t})\right) +\eps\cdot w(F^\ast_0) \leq
(1+2\eps)\beta\cdot w(F^\ast_0)\,,
\end{align*}
which proves the claim.
\end{proof}

Note that in case the given $p$ is smaller than the number of Steiner vertices in $F^\ast_0$,
or $c$ is smaller than the number of connected components in $F^\ast_0$, the
algorithm still outputs a Steiner forest, but the approximation factor may be
arbitrary. Finally, we provide proofs of \cref{thm:SF,thm:PSAKS-SF}.

\begin{proof}[Proof of \cref{thm:SF}]
Obtaining an \FPT algorithm for \SF parameterized by the number of terminals and 
connected components is not hard given an \FPT algorithm as the one given 
in~\cite{fuchs2007dynamic} for \ST: We only need to guess the sets of terminals 
that form connected components in the optimum Steiner forest. We can then invoke 
the algorithm of~\cite{fuchs2007dynamic} on each subset to compute an optimum 
Steiner tree connecting it. To bound the number of partitions of the terminal 
set, recall that the input to our algorithm has an integer $c$ upper-bounding 
the number of components in a solution with which we compare our solution. Thus, 
each terminal can be in one of at most $c$ components, so there are at most 
$c^{|R|}$ partitions of the terminal set $R$ that need to be considered. The 
algorithm of~\cite{fuchs2007dynamic} runs in time $(2 + \delta)^{|R|}\polyn$ for 
any constant $\delta > 0$, and this results in an algorithm with runtime $((2 + 
\delta)c)^{|R|}\polyn$ to solve \SF. We run this algorithm on the \SF instance 
that our preprocessing algorithm of \cref{lem:weightedPreprocessing} computes, 
in order to obtain \cref{thm:SF}.
\end{proof}

To obtain \cref{thm:PSAKS-SF} on lossy kernels, we rely on the fact that a 
\PSAKS exists for \ST parameterized by the number of terminals. It is known that 
despite being \FPT~\cite{DBLP:journals/networks/DreyfusW71,fuchs2007dynamic}, 
this problem does not admit polynomial (exact) kernels~\cite{dom}, unless 
$\mathsf{NP}\subseteq\mathsf{coNP}/\mathsf{Poly}$. However, as shown 
by~\citet{lokshtanov2017lossy}, the \mbox{\citet{borchers-du}} Theorem can be 
reinterpreted to show that a \PSAKS exists. Obtaining a \PSAKS for \SF can be 
done in essentially the same way as described in~\cite{lokshtanov2017lossy}, 
and together with \cref{lem:weightedPreprocessing} this gives a \PSAKS for our 
choice of parameters.

\begin{proof}[Proof of \cref{thm:PSAKS-SF}]
The \citet{borchers-du} Theorem states that for any optimum Steiner tree $T$ on 
terminal set $R$ there exists a collection of trees $T_1,\ldots,T_k$, such that 
all leaves of each tree belong to $R$, each $T_i$ contains 
$2^{\bigO{1/\eps}}$ terminals of~$R$, and the union $\bigcup_{i=1}^k T_i$ 
is a $(1+\eps)$-approximation of~$T$. This theorem can also be applied to each 
tree in the optimum \SF solution, since each such tree must be an optimum 
Steiner tree for its contained terminal set.

In particular, to compute a kernel, first we take the metric closure of the 
graph with $\tau$ terminals computed by our algorithm, so that any minimum cost 
tree connecting~$2^{\bigO{1/\eps}}$ terminals can be assumed to only contain 
$2^{\bigO{1/\eps}}$ Steiner vertices as well. We then compute an optimum Steiner 
tree for each subset of~$R$ of size at most~$2^{\bigO{1/\eps}}$. This is done 
using an FPT algorithm parameterized by the number of terminals, which takes 
polynomial time if $\eps$ is a constant. Within the union of all computed 
Steiner trees exists a $(1+\eps)$\hy{}approximate Steiner forest due to the 
\mbox{\citet{borchers-du}} Theorem, and the total number of vertices in this 
union is~$|R|^{2^{\bigO{1/\eps}}}$. However, the union is not of polynomial size 
in $|R|$ yet, due to the edge lengths. \citet{lokshtanov2017lossy} show that it 
is possible to round the edge lengths in such a way that the cost of every 
Steiner tree grows by at most a factor of $(1+\eps)$, and the edge lengths can 
be encoded using at most $\bigO{\log(|R|)+\log(1/\eps)}$ bits. For this an 
estimate on the cost of the optimum solution is needed, which can be obtained 
using the polynomial time $2$-approximation algorithm for \SF by 
\citet{DBLP:journals/siamcomp/AgrawalKR95}.

The number of terminals in the instance that we obtain after exhaustively 
applying our contractions is bounded in terms of our parameters $p$, $c$, and 
$\eps$ by \cref{lem:weightedPreprocessing}. Hence, the union of all computed 
solutions for terminal sets of size at most $2^{\bigO{1/\eps}}$ with rounded 
edge lengths is a polynomial-sized $(1+\eps)$-approximate kernel for \SF.
\end{proof}

\section{The unweighted directed Steiner tree problem}
\label{sec:unwDirST}

In this section we provide an \EPAS for the \unwDirST problem, in which each arc
has unit weight.

\optprobA
{\unwDirST}
{A directed graph $G = (V,A)$, and a set $R$ of terminals with a root terminal $r$.}
{A Steiner arborescence $T \subseteq G$ containing a directed path from $r$ to each terminal $v\in R$.}

\smallskip

The idea behind our algorithm given in this section is to
reduce the number of terminals of the input instance via a set of reduction
rules. That is, we would like to reduce the input graph $G$ to a graph $H$, and
prove that the number of terminals in $H$ is bounded by a function of our
parameter $p$ and the approximation ratio~$(1+\varepsilon)$. On $H$ we use the 
algorithm of \citet{DBLP:conf/icalp/Nederlof09} to obtain an optimum solution.

Our first reduction rule represents the idea that a terminal in the immediate
neighborhood of the root can be contracted to the root.
Observe that in this case our algorithm has to pay~$1$ for connecting such a
terminal to the root, however, any feasible solution must connect this terminal
as well using at least one arc---this argument is formalized in
\cref{lem:safenessForRootNeighbors} (cf.~\cref{sec:kernel-prelims}).

\begin{reductionRule}{R1}\label{rr:rootNeighbors}
If there is an arc from the root $r$ to a terminal $v\in R$, we contract the
arc~$(r,v)$, and declare the resulting vertex the new root.
\end{reductionRule}
\begin{lem}\label{lem:safenessForRootNeighbors}
Reduction Rule~\ref{rr:rootNeighbors} is $1$-safe and can be implemented in polynomial time.
Furthermore, there is a solution lifting algorithm running in polynomial time
and returning a Steiner arborescence if it gets a Steiner arborescence of the reduced graph as input.
\end{lem}
\begin{proof}
The implementation of the reduction rule is straightforward.
Let $H$ be a graph resulting from $G$ after the contraction of the arc $(r,v)$ to the new root $r'$, let
$T^*_H$ and $T^*_G$ denote optimal Steiner arborescences for $H$ and $G$,
respectively, and let $T_H$ be a Steiner arborescence in $H$.

Our solution lifting algorithm constructs a Steiner arborescence $T_G$ in G by
simply taking~$T_H$ and uncontracting $(r,v)$ in it. Note that $T_G$ spans all
terminals, as $T_H$ does in $H$ and we added $(r,v)$. Also, $T_G$ is an
arborescence, since $r$ has in-degree zero (as $r'$ has), $v$ has in-degree one,
and $T_G$ is clearly a tree. Thus $T_G$ is a Steiner arborescence in $G$.

The solution lifting algorithm adds 1 to the solution value, so that $w(T_G) = w(T_H) + 1$.
Note that $w(T^*_G) \ge w(T^*_H) + 1$ as the optimal solution in $G$ must
additionally connect $v$ to~$r$, i.e., it has to add some arc of cost 1.
Finally we have
\[
\frac{w(T_G)}{w(T^*_G)} \leq \frac{w(T_H) + 1}{w(T^*_H) + 1} \leq \max \left\{
\frac{w(T_H)}{w(T^*_H)} ; \frac{1}{1}\right\},
\]
so that if $T_H$ is a $\beta$-approximation of $T^*_H$, then $T_G$ is a
$(\max\{1;\beta\})$-approximation of $T^*_G$. Hence, the rule is $1$-safe.
\end{proof}

The idea behind our next reduction rule is the following.
Assume there is a Steiner vertex~$s$ in the optimum arborescence $T$ connected
to many terminals with paths not containing any other Steiner vertices. We can
then afford to buy all these paths emanating from $s$ together with a path
connecting the root to $s$. Formally, we say that a vertex $u$ is a
\emph{$k$-extended neighbor} of some vertex $v$, if there exists a directed path
$P$ starting in $v$ and ending in $u$, such that $V(P)\setminus\{v\}$ contains
at most $k$ Steiner vertices. Note that a vertex is always a $k$-extended
neighbor of itself for any $k$, and that each of the above terminals connected
to $s$ in $T$ is a $0$-extended neighbor of $s$.
We denote by $\ExtNei{k}(v)$ the set of all $k$-extended neighbors of $v$, and call it
the \emph{$k$-extended neighborhood} of $v$ (see \cref{fig:ExtNei}). By the
following observation, the Steiner vertex $s$ of $T$ lies in the $p$-extended
neighborhood of the root $r$. Therefore, there is a path containing at most $p$
Steiner vertices connecting $r$ to~$s$.

\begin{figure}[bt]
	\begin{minipage}{.6\textwidth}
		\begin{center}
			\usetikzlibrary{positioning,calc}

\begin{tikzpicture}
\tikzstyle{terminal}=[draw, , fill, black]
\tikzstyle{steiner}=[draw, black!30, circle , fill, text=black, inner sep=2pt]
\tikzstyle{steinerFull}=[draw, black!30, circle , fill, text=black, minimum width=.3cm]
\tikzstyle{steinerSel}=[draw, black, circle , fill, text=white, inner sep=3pt]
\tikzstyle{arrow}=[->,thick,>=stealth]
\tikzstyle{noArrow}=[->,thick,>=stealth, dotted]
\tikzstyle{vertex}=[draw, thick, circle, black!30, text = black, fill, inner sep=2pt]
\tikzstyle{vertexSmall}=[draw, thick, circle, black!30, text = black, fill, inner sep=3pt]
\tikzstyle{vertexSel}=[draw, thick, circle, inner sep=3pt, black, fill, text = white]

\begin{scope}[local bounding box=ext0]
\node[steiner] at (3,2) (s) {$s$};

\node[terminal] at (0, 1) (t11) {};
\node[terminal] at (2, 1) (t12) {};
\node[terminal] at (3, 1) (t13) {};
\node[terminal] at (4, 1) (t14) {};
\node[terminal] at (6, 1) (t15) {};
\node[terminal] at (0, 0) (t02) {};
\node[terminal] at (2, 0) (t03) {};
\node[terminal] at (6, 0) (t04) {};

\node[terminal] at (1, 0) (t011) {};
\node[terminal] at (5, 0) (t012) {};

\node[steinerFull] at (1,1) (s1) {};
\node[steinerFull] at (5,1) (s2) {};
\node[steinerFull] at (3,0) (s3) {};
\node[steinerFull] at (4,0) (s4) {};

\foreach \i in {1,...,5} {
  \draw[arrow] (s) -- (t1\i);
}

\draw[arrow] (t11) -- (t02);
\draw[arrow] (t12) -- (t03);
\draw[arrow] (t15) -- (t04);

\draw[noArrow] (s) -- (s1);
\draw[noArrow] (s) -- (s2);
\draw[noArrow] (t12) -- (s3);
\draw[noArrow] (t14) -- (s3);
\draw[noArrow] (t14) -- (s4);

\draw[noArrow] (s1) -- (t011);
\draw[noArrow] (s2) -- (t012);
\end{scope}

\begin{scope}[local bounding box=ext1, shift={($(ext0) - (3,4)$)}]
\node[steiner] at (3,2) (s) {$s$};

\node[terminal] at (0, 1) (t11) {};
\node[terminal] at (2, 1) (t12) {};
\node[terminal] at (3, 1) (t13) {};
\node[terminal] at (4, 1) (t14) {};
\node[terminal] at (6, 1) (t15) {};
\node[terminal] at (0, 0) (t02) {};
\node[terminal] at (2, 0) (t03) {};
\node[terminal] at (6, 0) (t04) {};

\node[terminal] at (1, 0) (t011) {};
\node[terminal] at (5, 0) (t012) {};

\node[steinerFull] at (1,1) (s1) {};
\node[steinerFull] at (5,1) (s2) {};
\node[steinerFull] at (3,0) (s3) {};
\node[steinerFull] at (4,0) (s4) {};

\foreach \i in {1,...,5} {
  \draw[arrow] (s) -- (t1\i);
}

\draw[arrow] (t11) -- (t02);
\draw[arrow] (t12) -- (t03);
\draw[arrow] (t15) -- (t04);

\draw[arrow] (s) -- (s1);
\draw[arrow] (s) -- (s2);
\draw[arrow] (t12) -- (s3);
\draw[arrow] (t14) -- (s3);
\draw[arrow] (t14) -- (s4);

\draw[arrow] (s1) -- (t011);
\draw[arrow] (s2) -- (t012);
\end{scope}
\end{tikzpicture}
		\end{center}
	\end{minipage}
  \begin{minipage}{.35\textwidth}
		\caption{
	An example of extended neighborhood of Steiner vertex $s$.
	The set $\ExtNei{0}(s)$ is depicted on the top using full arcs, while the
	vertices connected by dotted arcs are not a part of this set.
	The set $\ExtNei{1}(s)$ is depicted on the bottom using full arcs.
	}\label{fig:ExtNei}
  \end{minipage}
\end{figure}

\begin{obs}\label{obs:shortDipath}
Let $G = (V, A)$ be a directed graph with root $r\in R$. Suppose there
exists a Steiner arborescence $T\subseteq G$ with at most $p$ Steiner vertices.
It follows that $V(T) \subseteq \ExtNei{p}(r)$.
\end{obs}

In what follows we fix $\varepsilon > 0$.
The second reduction rule contracts a path from $r$ to a Steiner vertex $s$ in the $p$-extended neighborhood of $r$
together with the $0$-extended neighborhood of $s$ if this neighborhood is sufficiently large.

\begin{reductionRule}{R2}\label{rr:largeStarPathContraction}
If there exists a Steiner vertex $s$ with $\left|\ExtNei{0}(s)\right| \ge
p / \varepsilon$ and \mbox{$s\in\ExtNei{p}(r)$,} so that there is an $r\to s$
path $P$ containing at most $p$ Steiner vertices, then we contract the subgraph
of $G$ induced by $\ExtNei{0}(s)$ and $P$ in $G$, and declare the resulting
vertex the new root.
\end{reductionRule}

\begin{lem}\label{lem:safenessForLargeStarPathContraction}
Reduction Rule~\ref{rr:largeStarPathContraction} is $(1 + \varepsilon)$-safe and can be implemented in polynomial time.
Furthermore, there is a solution lifting algorithm running in polynomial time
and returning a Steiner arborescence if it gets a Steiner arborescence of the reduced graph as input.
\end{lem}
\begin{proof}
Checking the applicability of Rule~\ref{rr:largeStarPathContraction} and finding $s$ together with $\ExtNei{0}(s)$ can
be done in polynomial time as follows. We set arc lengths so that
each arc ending at a terminal has length zero, while arcs ending at Steiner vertices
have length one. Now a length of a directed path $P$ from the root
corresponds to the number of Steiner vertices in~$P$.
Then, we run an algorithm for finding a shortest path from $r$ to each vertex
which allows us to find the set $\ExtNei{p}(r)$. Finally, for each $s\in \ExtNei{p}(r)$
we compute $\ExtNei{0}(s)$ by a simple breadth-first search.

We now specify the solution lifting algorithm.
Denote by $H$ the reduced graph obtained from $G$ by applying
\ref{rr:largeStarPathContraction}.
Let $T_H$ be a solution of the reduced instance $H$ and let $T^*_H$ be an optimal solution in $H$.
Consider the graph $Q$, which is the union of $P$ and the subgraph of $G$
induced by $\ExtNei{0}(s)$. The solution lifting algorithm first computes an
arborescence $A$ of $Q$ rooted in $r$ (e.g., by a depth-first search).
Define $T_G$ as the union of $T_H$ and $A$. We show that $T_G$ is a 
Steiner arborescence.

First, observe that $T_G$ spans all terminals as $T_H$ contains all terminals in $H$
and $A$ is an arborescence containing all vertices in $Q$. Note that $T_G$ is a tree as $A$ is an arborescence of $Q$,
$T_H$ is a tree, and $T_H$ contains at most one arc from the root in $H$ to each vertex
(recall that the root in $H$ was created by contracting $\ExtNei{0}(s)\cup V(P)$).
The root in $T_G$ has clearly in-degree zero, while all other vertices have in-degree one,
since this holds for $H$ as $T_H$ is an arborescence, and $A$ is an arborescence of $Q$ rooted in $r$.
Thus $T_G$ is a Steiner arborescence in $G$.

It remains to show the safeness of the rule. Let $x$ be the total number of
terminals in $\ExtNei{0}(s)\cup V(P)$ (not counting the root) and let~$T^*_G$ be an optimal solution in~$G$. Note that $w(T_G)
\le w(T_H) + x + p$. We obtain a solution for $H$ of weight at most $w(T^*_G) -
x$ by starting with $T^*_G$, removing $x$ arcs each having one of the $x$
non\hy{}root terminals in $\ExtNei{0}(s)\cup V(P)$ (and thus not in~$H$) as
their head, identifying all vertices in $\ExtNei{0}(s)\cup V(P)$ with the new
root, and removing loops and parallel arcs. Thus $w(T^*_G) \ge w(T^*_H) + x$ and
we get
\[
\frac{w(T_G)}{w(T^*_G)} \leq \frac{w(T_H) + x + p}{w(T^*_H) + x} \leq
\max \left\{\frac{w(T_H)}{w(T^*_H)}; \frac{x+p}{x} \right\} \leq \max
\left\{\frac{w(T_H)}{w(T^*_H)}; 1 + \varepsilon \right\}.
\]
The last inequality is valid because $x \ge p / \varepsilon$. Thus if
$T_H$ is a $\beta$-approximation of $T^*_H$, then $T_G$ is a
$(\max\{1+\eps;\beta\})$-approximation of $T^*_G$, and so the reduction rule is
$(1+\eps)$-safe.
\end{proof}

Now we prove that if none of the above reduction rules is applicable and our
algorithm was provided with a correct value for parameter $p$, then the number
of terminals in the reduced graph can be bounded by $p^2/\varepsilon$.

\begin{lem}\label{lem:smallInstanceAfterRR}
Let $G$ be an instance of \dirST, and denote by $H$ the graph obtained from $G$
by exhaustive application of Reduction Rules~\ref{rr:rootNeighbors} and
\ref{rr:largeStarPathContraction}.
Suppose that there exists a Steiner arborescence in~$G$ containing at most $p$
Steiner vertices.
It follows that the remaining terminal set $R$ of $H$ has size less than
$p^2/\varepsilon$.
\end{lem}
\begin{proof}
Observe first that both our reduction rules use contractions in the underlying graph and thus if there was a solution $T^*_G$ in $G$ with at most $p$ Steiner vertices, then there is a solution $T^*_H$ in $H$ again containing at most $p$ Steiner vertices.

Since Reduction Rule~\ref{rr:rootNeighbors} is not applicable to $H$, we
conclude that $\ExtNei{0}(r) \cap R = \emptyset$.
As Reduction Rule~\ref{rr:largeStarPathContraction} is not applicable to $H$, it
holds that $\left|\ExtNei{0}(s) \cap R \right| < p / \varepsilon$ for
every Steiner vertex $s\in\ExtNei{p}(r)$.
Therefore, $|R| < p^2 / \varepsilon$, since any terminal in $H$ must be
in the 0-extended neighborhood of some Steiner vertex in $T^*_H$
and there are at most $p$ Steiner vertices in $T^*_H$.
\end{proof}

The last step of the algorithm (cf.\ proof of \cref{thm:unwDirST})
is to compute an optimum solution in the graph
$H$ obtained from the input graph $G$ after exhaustively applying Reduction Rules~\ref{rr:rootNeighbors} and
\ref{rr:largeStarPathContraction}. From the resulting arborescence in $H$, we obtain an arborescence in $G$
by running the solution lifting algorithms for each reduction rule applied (in the reverse order);
the existence and correctness of the solution lifting algorithms for our reduction rules
is provided by
\cref{lem:safenessForRootNeighbors,lem:safenessForLargeStarPathContraction}.
The algorithm is summarized in~\cref{alg:unweightedDirectedST}.

\begin{algorithm}[h!]
\setcounter{AlgoLine}{0}
\SetKwInOut{Input}{input}\SetKwInOut{Output}{output}

\Input{directed graph $G = (V, A)$, terminals $R\subseteq V$, root $r\in R$,
and integer $p$}
\Output{Steiner arborescence $T\subseteq G$, if $p$ is at most the nr.\ of
terminals in the optimum}

\BlankLine
\If(\tcc*[h]{no solution with at most $p$ Steiner vertices}){$R \setminus
\ExtNei{p}(r) \neq \emptyset$} {
  \Return ``no''\;
  }
\While{Reduction Rule~\ref{rr:rootNeighbors}
or~\ref{rr:largeStarPathContraction} is applicable}{
\If(\tcc*[h]{Reduction Rule~\ref{rr:rootNeighbors}}){there is an arc from $r$
to
$v\in R$} {
	Contract the arc~$(r,v)$, and declare the resulting vertex the new root.
}
\If(\tcc*[h]{Reduction Rule~\ref{rr:largeStarPathContraction}}){\mbox{there
exists $s\in V\setminus R$ with $s\in\ExtNei{p}(r)$ and
$\left|\ExtNei{0}(s)\right| \ge p / \epsilon$}} {
	Find an $r\to s$ path $P$ with at most $p$ Steiner vertices.
	Contract the subgraph of $G$ induced by $\ExtNei{0}(s)$ and $P$, and
declare the resulting vertex
the new root.
}
}

\If(\tcc*[h]{no solution with at most $p$ Steiner vertices}){$|R| > p^2 /
\varepsilon$} {
  \Return ``no''\;
}
Run the \FPT algorithm of~\cite{DBLP:conf/icalp/Nederlof09}; let $T$ be the
returned solution.\nllabel{algDir:FPTdirected}

In the reverse order of application of Reduction Rules~\ref{rr:rootNeighbors}
and~\ref{rr:largeStarPathContraction}:

\qquad Revert the contraction of the reduction rule.

\qquad Run the solution lifting algorithm for the reduction rule on $T$.

\qquad Store the resulting arborescence in $T$.

\Return $T$\;

\caption{Algorithm for solving \dirST.
As explained earlier, all steps except line~\ref{algDir:FPTdirected} can be implemented in
polynomial time.
}
\label{alg:unweightedDirectedST}
\end{algorithm}

\begin{proof}[Proof of \cref{thm:unwDirST}]
If neither Reduction Rule~\ref{rr:rootNeighbors} nor
\ref{rr:largeStarPathContraction} is applicable and the current number of
terminals exceeds the bound $p^2 / \varepsilon$ we can return ``no''
as it follows from \cref{lem:smallInstanceAfterRR} that no optimal solution
with at most $p$ Steiner vertices exists.
If this is not the case we return an optimal solution using the
algorithm of~\cite{DBLP:conf/icalp/Nederlof09}, which runs in time $2^{|R|}
\polyn$
where $R$ is the current set of terminals with size at most $p^2 / \varepsilon$.
As explained earlier both reduction rules can be implemented in polynomial time,
together with their solution lifting algorithms.
Thus the total running time is $2^{p^2/\eps} \polyn$.
The approximation guarantee and correctness of the obtained solution follow from
\cref{lem:safenessForRootNeighbors,lem:safenessForLargeStarPathContraction}.
\end{proof}

\section{The weighted directed Steiner tree problem}

Here, we prove that the standard reduction from the \domSet problem to the \dirST
problem (with arc weights) translates into inapproximability of the latter
problem. By a recent result of~\citet{dom-set}, there is no $f(b)$-approximation
algorithm for the \domSet problem, even when parameterizing by the size $b$ of the
optimum solution, unless $\W{1} = \FPT$.

\optprobA{\domSet}
{an undirected graph $G = (V,E)$.}
{the smallest \emph{dominating set} $U\subseteq V$ for which every $v\in V$
either is in $U$ or $v$ has a neighbor in $U$.}

\begin{figure}[bt]
  \centering
  \usetikzlibrary{calc}

\begin{tikzpicture}
\tikzstyle{terminal}=[draw, , fill, black]
\tikzstyle{steiner}=[draw, black!30, circle , fill, text=black, inner sep=2pt]
\tikzstyle{steinerSel}=[draw, black, circle , fill, text=white, inner sep=3pt]
\tikzstyle{arrow}=[<-,thick,>=stealth]
\tikzstyle{vertex}=[draw, thick, circle, black!30, text = black, fill, inner sep=2pt]
\tikzstyle{vertexSmall}=[draw, thick, circle, black!30, text = black, fill, inner sep=3pt]
\tikzstyle{vertexSel}=[draw, thick, circle, inner sep=3pt, black, fill, text = white]
\tikzstyle{graphEdge}=[ultra thick]

\def\dist{1.2}

\begin{scope}[node distance=1.4cm, local bounding box=graph]
\node[vertexSmall] (v1) {a};
\node[vertex, above left of=v1] (v2) {b};
\node[vertexSel, above of=v2] (v3) {c};
\node[vertex, above right of=v3] (v4) {d};
\node[vertexSel, below right of=v4] (v5) {e};
\node[vertex, below of=v5] (v6) {f};

\foreach \u/\v in {1/2, 1/4, 1/5, 2/3, 2/5, 2/6, 3/4, 4/5, 5/6} {
  \draw[graphEdge] (v\u) -- (v\v);
}
\end{scope}

\begin{scope}[shift={($ (graph) + (4.5,-2) $)}]
\node[terminal, label={270:$a$}] (T1) {};
\foreach \x/\labx [remember=\x as \px (initially 1)] in {2/b, 3/c, 4/d, 5/e, 6/f} { 
  \node[terminal, label={270:$\labx$}] at ($(T\px) +(\dist,0)$) (T\x) {};
}

\node[steiner] at ($(T1) + (0,2.5)$) (S1) {a};
\foreach \i/\type/\letter [remember=\i as \pi (initially 1)] in {2//b, 3/Sel/c, 4//d, 5/Sel/e, 6//f} {
  \node[steiner\type] at ($(S\pi) + (\dist,0)$) (S\i) {\letter};
}

\foreach \u/\v in {1/1, 2/2, 3/3, 4/4, 5/5, 6/6, 1/2, 2/1, 1/4, 4/1, 1/5, 5/1, 2/3, 3/2, 2/5, 5/2, 2/6, 6/2, 3/4, 4/3, 4/5, 5/4, 5/6, 6/5} {
  \draw[arrow] (T\u) -- (S\v);
}

\node[terminal, label={90:root}] (root) at ($(S3) + (0,\dist)$) {};
\foreach \i in {1, ..., 6} {
  \draw[arrow] (S\i) -- (root);
}

\node at ($(S6) + (.7*\dist,0) + (0,.7*\dist)$)  (W1) {1};
\node at ($(W1) - (0,1.8*\dist)$) {0};

\end{scope}

\end{tikzpicture}
  \caption{An example for the reduction. A~graph $G$ with its dominating
set $U=\{c,e\}$ on the left. The corresponding instance of \dirST to the right.}
  \label{fig:domSetReduction}
\end{figure}

\begin{proof}[Proof of \cref{thm:dirST}]
We give a parameterized reduction from the \domSet problem parameterized by the
size of the solution $U$, which we denote by $b=|U|$.

For an overview of the reduction please refer to \cref{fig:domSetReduction}. Let
$G = (V,E)$ be a graph in which we are searching for the smallest dominating
set of size $b$ and let $n = |V|$ and $m = |E|$. We create an instance of \dirST
having $2n + 1$ vertices and $n + 2m$ arcs as follows. There are $n$ terminals,
each corresponding to a vertex in $V$, one auxiliary terminal (the root), and
$n$ Steiner vertices again corresponding to vertices in $V$. There are arcs of
two kinds. The first kind of arcs are of weight $1$ and connect the root to
each Steiner vertex, i.e., they are directed towards the Steiner vertices. The
second kind of arcs are of weight $0$ and connect the Steiner vertices with
the terminals, directed towards the terminals. There is an arc from each Steiner
vertex corresponding to a vertex $w\in V$ to every terminal corresponding to a
vertex $v\in V$ if $v = w$ or $v$ is a neighbor of $w$ in $G$.

Observe that there is a dominating set of size $b$ in $G$ if and only if there
is an arborescence connecting the root to all terminals of cost $b$. Note also
that this arborescence contains $b$ Steiner vertices. Thus we set the parameter
$p$ to value $b$.

Suppose that there is a parameterized $f(p)$-approximation algorithm for the
\dirST problem for parameter $p$ and a computable function $f$. Then, we
would obtain a parameterized $f(b)$-approximation algorithm for the \domSet
problem parameterized by the size $b$ of the solution. This would imply $\W{1}
= \FPT$ by~\cite{dom-set}.
\end{proof}

\section{Refuting a PSAKS for \unwDirST}
In this section, we prove that the \unwDirST problem does not admit a $(2 -
\varepsilon)$-approximate polynomial kernel for any constant $\varepsilon > 0$
unless $\NP \subseteq \coNPp$.
We use a framework for proving lower bounds on approximate polynomial kernels by~\citet{lokshtanov2017lossy} and present an $\alpha$-gap cross composition (for $\alpha = 2 - \varepsilon$).
For the composition, we need to define a polynomial equivalence.
\begin{dfn}
An equivalence relation $\equiv$ on $\Sigma^*$, where $\Sigma$ is a finite
alphabet, is called a \emph{polynomial equivalence relation} if
 \begin{enumerate}
  \item The equivalence of any $x, y \in \Sigma^*$ can be checked in time
polynomial in $|x| + |y|$.
  \item Any finite set $S \subseteq \Sigma^*$ has at most $(\max_{x \in S}
|x|)^{O(1)}$ equivalence classes.
 \end{enumerate}
\end{dfn}

Now we explain how the composition works.
Let $L \subseteq \Sigma^*$ be a language, and let $x_1,\dots,x_t \in \Sigma^*$
be strings belonging to the same class of some polynomial equivalence $\equiv$.
The composition, given $x_1,\dots x_t$, runs in time polynomial in $\sum_{i
=1}^t |x_i|$ and computes $c \in \R$ and an instance $(G,R,p)$ of the \unwDirST
problem parameterized by $p$  such that:
\begin{enumerate}
 \item If $x_i \in L$ for some $1 \leq i \leq t$, then $G$ contains a Steiner arborescence containing at most $c$ arcs.
 \item If $x_i \not \in L$ for all $1 \leq i \leq t$, then any Steiner arborescene of $G$ contains at least $\alpha \cdot c$ arcs.
 \item The parameter $p$ is bounded by a polynomial in $\log t + \max_{1 \leq i \leq t}|x_i|$.
\end{enumerate}
By the framework of \citet{lokshtanov2017lossy}, if $L$ is an $\NP$-hard
language, then the \unwDirST problem does not admit a polynomial-sized
$\alpha$\hy{}approximate kernel for parameter $p$, unless $\NP \subseteq
\coNPp$. We use the \setCover problem as the language~$L$.

\optprobA{\setCover} {A universe $U$, a set ${\cal P}$ of subsets of $U$, and a positive integer $b$.}
{A set ${\cal C} \subseteq {\cal P}$ such that $|{\cal C}| \leq b$ and $U =
\bigcup_{C
\in {\cal C}} C$.}

\smallskip

We call $b$ the \emph{budget}. Let ${\cal I}_1,\dots,{\cal I}_t$ be instances of the \setCover problem.
We define the polynomial equivalence $\equiv$ as follows.
Two \setCover instances $(U_1, {\cal P}_1, b_1)$ and $(U_2, {\cal P}_2, b_2)$ are equivalent in $\equiv$ if $|U_1| = |U_2| = n$, $|{\cal P}_1| = |{\cal P}_2| = m$ and $b_1 = b_2 = b$.
Thus, we can suppose that all instances ${\cal I}_1,\dots,{\cal I}_t$ are over the same universe $U$.
It is straightforward to verify that the relation $\equiv$ is a polynomial equivalence relation.

We can also suppose that $m$ is polynomial in $n$ and either each instance
${\cal I}_i$ has a set cover of size at most $b$ or each set cover has size at
least $\gamma b$ for arbitrary constant $\gamma$ (actually $\gamma$ can be
$O(\log n)$ but we do not need this here). By a result of
\citet{dinur2014setcover}, the \setCover problem is still $\NP$-hard in this
case.

The first step of the $\alpha$-gap cross composition is to convert each instance ${\cal I}_k$ to an instance $G_k$ of the \unwDirST problem.
The construction is similar to the reduction in the proof of~\cref{thm:dirST}.
Let ${\cal I}_k = (U, {\cal P}_k, b)$ be an instance of \textsc{Set Cover}. We
create a terminal root $r_k$ and for each $S^k_j \in {\cal P}_k$ we create a Steiner
vertex $s^k_j$.
We add a directed path of length $n$ from $r_k$ to each $s^k_j$.
Then we create a terminal $t^k_i$ for every $i \in U$ and create the incidence
graph of $U$ and~${\cal P}$, i.e., we add an arc $(s^k_j, t^k_i)$ if $i \in S^k_j$
where $S^k_j$ is the set in ${\cal P}$ corresponding to~$s^k_j$.

If ${\cal I}_k$ is a {\it yes}-instance (it has a set cover of size at most $b$), then $G_k$ has a Steiner arborescence with $bn + n = (b + 1)n$ arcs.
On the other hand, if ${\cal I}_k$ is a {\it no}-instance (each set cover has size at least $\gamma b$), then each Steiner arborescence of $G_k$ has more than $\gamma bn + n$ arcs.

Now we combine all $G_k$ into one instance $G$ of the \unwDirST problem.
First, we create a root $r$ of the instance $G$.
We connect $r$ and each vertex $r_k$ (the root of the instance $G_k$) by a directed path $P_k$ of length $d$ (the value of $d$ will be determined later).
Thus, the root $r$ has degree $t$.
Finally, we identify all terminals $t^k_i$ of all graphs $G_k$ corresponding to the
same element $i$ in $U$ into one terminal $t_i$, i.e., the graph $G$ has $n$ terminals
apart from the root.
See \cref{fig:gap_composition} for a sketch of the composition.

\begin{figure}[bt]
\centering
  \includegraphics[width=\textwidth]{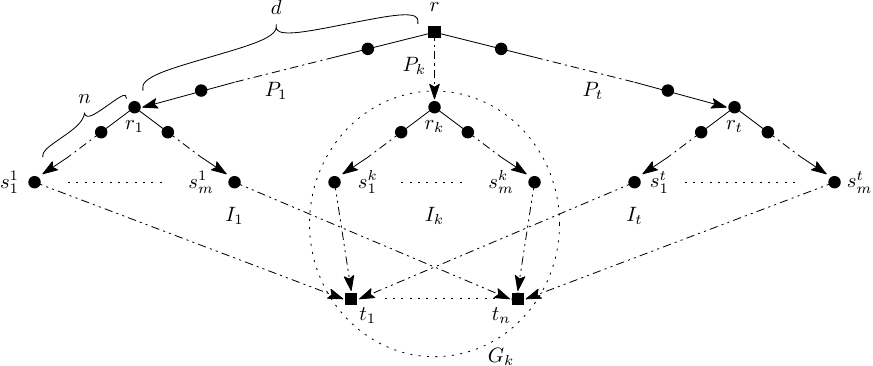}
  \caption{
  Sketch of the $(2-\varepsilon)$-gap cross composition.
  All arcs are oriented in ``top-down'' direction from the root $r$ to terminals $t_i$.
  The graph $I_k$ is an incident graph of the instance ${\cal I}_k = (U, {\cal
P}_k, b)$ of the \setCover problem.
  The graph $G_k$ is the graph $I_k$ with paths from the vertex $r_k$ to the vertices $s^k_1,\dots,s^k_m$.
  }
  \label{fig:gap_composition}
\end{figure}

\begin{lem}
\label{lem:yesInst}
If for some $k$, ${\cal I}_k$ is a {\it yes}-instance then $G$ has a Steiner arborescence with at most $d + (b+1)n$ arcs.
\end{lem}
\begin{proof}
Let ${\cal C}$ be a set cover of ${\cal I}_k$ of size at most $b$.
The arborescence $T$ contains the path $P_k$ from $r$ to $r_k$, thus it contains $d$ arcs.
Let $S$ be Steiner vertices in $G_k$ corresponding to the sets in ${\cal C}$.
We add to $T$ all the paths from $r_k$ to Steiner vertices in $S$;
as $|S|\le |{\cal C}| \le b$, these paths have at most  $bn$ arcs.
Since ${\cal C}$ is a set cover, there are $n$ arcs from Steiner vertices in $S$ to all terminals of $G$ and we add them to $T$.
Thus, $T$ connects the root $r$ of $G$ to all the terminals of $G$ and it has
$d + (b+1)n$ arcs.
\end{proof}

\begin{lem}
\label{lem:cheating}
Let $T$ be a Steiner arborescence of $G$.
Suppose $T$ contains two distinct paths $P_i$ and~$P_j$.
Then, $T$ has at least $2d + 3n$ arcs.
\end{lem}
\begin{proof}
The paths $P_i$ and $P_j$ are edge disjoint and each contains $d$ arcs.
Further, $T$ contains at least $n$ arcs from $r_i$ to some Steiner vertex in $G_i$ and at least $n$ arcs from $r_j$ to some Steiner vertex in $G_j$.
Finally, we still need $n$ arcs to connect the terminals and we get
 \[
  |E(T)| \geq 2d + 3n.
 \]
\end{proof}

\begin{lem}
\label{lem:honest}
Let $T$ be a Steiner arborescence of $G$ such that $T$ contains only one path $P_k$.
If all instances ${\cal I}_1,\dots,{\cal I}_t$ are {\it no}-instances, then any Steiner arborescence of $G$ has at least $d + n(\gamma b + 1)$ arcs.
\end{lem}
\begin{proof}
Let $T'$ be an arborescence we get from $T$ when we remove the path $P_k$.
Since $V(P_k) \cap V(T') = \{r_k\}$, the arborescence $T'$ is a Steiner arborescence of $G_k$.
Thus, the arborescence $T'$ has at least $n(\gamma b + 1)$ arcs, because the
instance ${\cal I}_k$ is a {\it no}-instance. Adding the $d$ edges of $P_k$, we
obtain the claimed bound.
\end{proof}

Now we calculate the value of $d$.
We set $d$ large enough so that Steiner arborescences which contain more than
one path $P_k$ are bigger than Steiner arborescences which contain only one
such path.
Formally, by the above two lemmas we want
\[
2d + 3n \geq d + n(\gamma b + 1).
\]
Thus, we set $d = n (\gamma b - 2)$ and we get the following corollary of \cref{lem:cheating} and \cref{lem:honest}.

\begin{crl}
\label{crl:noInst}
If all instances ${\cal I}_1,\dots,{\cal I}_t$ are {\it no}-instances, then each Steiner arborescence of $G$ has at least $n(2\gamma b - 1)$ arcs.
\end{crl}

\begin{obs}
The graph $G$ has a Steiner arborescence $T$ with at most $d + n^2 + n$ Steiner vertices.
\end{obs}
\begin{proof}
We take a path $P_k$ from $r$ into an arbitrary vertex $r_k$ (with $d$ Steiner
vertices) and an arbitrary Steiner arborescence in $G_k$ (with at most $n^2 + n$
Steiner vertices---from a trivial set cover when each element is covered by its
own set).
\end{proof}

Thus, our parameter $p$ of $G$ (the number of Steiner vertices in the optimum)
is bounded by a polynomial in $n$, as $d = n (\gamma b - 2)$ and $b\leq n$.
If there is a {\it yes}-instance among ${\cal I}_1,\dots,{\cal I}_t$, then by \cref{lem:yesInst} we know that the optimal Steiner arborescence of $G$ has at most $d + (b+1)n = n\bigl((\gamma + 1)b - 1\bigr)$ arcs.
If there are {\it no}-instances among ${\cal I}_1,\dots,{\cal I}_t$ only, then
by \cref{crl:noInst} the optimal Steiner arborescence of $G$ has at least $n(2
\gamma b - 1)$ arcs.
This means that
\[
 \frac{n(2\gamma b - 1)}{n\bigl((\gamma + 1)b - 1\bigr)} \geq 2 - \varepsilon
\]
for $\gamma$ large enough.
Thus, for any constant $\varepsilon > 0$ we created a $(2 - \varepsilon)$-gap cross composition from the \setCover problem to the \unwDirST
problem parameterized by the number of Steiner vertices in the optimum.
This refutes the existence of polynomial-sized $(2 - \varepsilon)$-approximate
kernels for this problem, unless $\NP \subseteq \coNPp$, and proves
\cref{thm:no-PSAKS-unwDirST}.

\section{Conclusions and open problems}
Recently, it was shown that contracting stars not only leads to parameterized
approximation schemes for the \ST problem, as outlined in this paper, but also
behaves well in practical computations~\cite{paceImplementation,
StarContractionHeuristic}. In fact, this idea was used as a heuristic, which
significantly improves approximations of minimum spanning trees. The
implementation of this idea together with only a few additional heuristics was
awarded the 4th place (out of 24) in the PACE challenge 2018 in the very
competitive Track~C~\cite{paceReport}.

From our theoretical work, we leave the following open problems:
\begin{itemize}
 \item The runtimes of our approximation schemes may be improvable. In
particular, we conjecture that a linear dependence on our parameter $p$ should
suffice in the exponent of both algorithms in \cref{thm:ST,thm:unwDirST}. It 
would also be very interesting to obtain runtime lower
bounds for our approximation schemes under some reasonable complexity 
assumption.
 \item Given that we obtain a \PSAKS for the \ST problem, but not for the
\unwDirST problem (even though we show an \EPAS for each of them), one
remaining open question is what the best approximation ratio obtainable by a
polynomial-sized kernel is for the latter. Namely, is there a 
polynomial-sized $\alpha$-approximate kernel for \unwDirST for some constant 
$\alpha\geq 2$?
\item As mentioned in \cref{sec:related}, a parameterized approximation scheme
and a \PSAKS exist for the \textsc{Bidirected Steiner Network} problem with
planar optimum~\cite{chitnis2017parameterized} for parameter~$|R|$. The \PSAKS
uses a generalization of the \PSAKS for \ST with parameter~$|R|$ by
\citet{lokshtanov2017lossy}. Hence, it is natural to ask whether or not this is also
the case for our parameter~$p$, i.e., whether or not there is a parameterized approximation
scheme and/or a \PSAKS for \textsc{Bidirected Steiner Network} with planar
optimum when parameterized by $p$.
\end{itemize}

\paragraph*{Acknowledgments}
We would like to thank Michael Lampis and \'Edouard Bonnet for helpful
discussions on the problem.
Also, we thank Ji\v{r}\'{i} Sgall and Martin B\"{o}hm for finding a mistake in
a preliminary version of the proof of \cref{lem:UB-numOfNotChargeable}.

\bibliographystyle{hplainnat} %
\bibliography{../src/papers,../src/lit}

\begin{thebibliography}{34}
\expandafter\ifx\csname url\endcsname\relax
  \def\url#1{\texttt{#1}}\fi
\expandafter\ifx\csname doi\endcsname\relax
  \def\doi#1{\burlalt{doi:#1}{http://dx.doi.org/#1}}\fi
\expandafter\ifx\csname urlprefix\endcsname\relax\def\urlprefix{URL }\fi
\expandafter\ifx\csname href\endcsname\relax
  \def\href#1#2{#2}\fi
\expandafter\ifx\csname burlalt\endcsname\relax
  \def\burlalt#1#2{\href{#2}{#1}}\fi
\expandafter\ifx\csname eprint\endcsname\relax
  \def\eprint#1{\burlalt{arXiv:#1}{http://arxiv.org/abs/#1}}\fi
\providecommand{\natexlab}[1]{#1}

\bibitem[Agrawal et~al.(1995)Agrawal, Klein, and
  Ravi]{DBLP:journals/siamcomp/AgrawalKR95}
Ajit Agrawal, Philip Klein, and R.~Ravi.
\newblock When trees collide: An approximation algorithm for the generalized
  steiner problem on networks.
\newblock \emph{{SIAM} Journal on Computing}, 24\penalty0 (3):\penalty0
  440--456, jun 1995, \doi{10.1137/s0097539792236237}.

\bibitem[Björklund et~al.(2007)Björklund, Husfeldt, Kaski, and
  Koivisto]{BjorklundHKK07}
Andreas Björklund, Thore Husfeldt, Petteri Kaski, and Mikko Koivisto.
\newblock Fourier meets möbius: fast subset convolution.
\newblock In \emph{Proceedings of the thirty-ninth annual {ACM} symposium on
  Theory of computing - {STOC} 07}. {ACM} Press, 2007,
  \doi{10.1145/1250790.1250801}.

\bibitem[Bonnet and Sikora(2019)]{paceReport}
{\'E}douard Bonnet and Florian Sikora.
\newblock {The PACE 2018 Parameterized Algorithms and Computational Experiments
  Challenge: The Third Iteration}.
\newblock In Christophe Paul and Michal Pilipczuk, editors, \emph{13th
  International Symposium on Parameterized and Exact Computation (IPEC 2018)},
  volume 115 of \emph{Leibniz International Proceedings in Informatics
  (LIPIcs)}, pages 26:1--26:15, Dagstuhl, Germany, 2019. Schloss
  Dagstuhl--Leibniz-Zentrum fuer Informatik.
\newblock ISBN 978-3-95977-084-2, \doi{10.4230/LIPIcs.IPEC.2018.26}.

\bibitem[Borchers and Du(1997)]{borchers-du}
Al~Borchers and Ding-Zhu Du.
\newblock {The $k$-{S}teiner Ratio in Graphs}.
\newblock \emph{SIAM Journal on Computing}, 26\penalty0 (3):\penalty0 857--869,
  1997, \doi{10.1137/S0097539795281086}.

\bibitem[Borradaile et~al.(2009)Borradaile, Klein, and
  Mathieu]{DBLP:journals/talg/BorradaileKM09}
Glencora Borradaile, Philip Klein, and Claire Mathieu.
\newblock {AnO}(nlogn) approximation scheme for steiner tree in planar graphs.
\newblock \emph{{ACM} Transactions on Algorithms}, 5\penalty0 (3):\penalty0
  1--31, jul 2009, \doi{10.1145/1541885.1541892}.

\bibitem[Byrka et~al.(2013)Byrka, Grandoni, Rothvoss, and
  Sanit{\`{a}}]{DBLP:journals/jacm/ByrkaGRS13}
Jaros{\l}aw Byrka, Fabrizio Grandoni, Thomas Rothvoss, and Laura Sanit{\`{a}}.
\newblock Steiner tree approximation via iterative randomized rounding.
\newblock \emph{Journal of the {ACM}}, 60\penalty0 (1):\penalty0 1--33, feb
  2013, \doi{10.1145/2432622.2432628}.

\bibitem[Charikar et~al.(1999)Charikar, Chekuri, yat Cheung, Dai, Goel, Guha,
  and Li]{DBLP:journals/jal/CharikarCCDGGL99}
Moses Charikar, Chandra Chekuri, To~yat Cheung, Zuo Dai, Ashish Goel, Sudipto
  Guha, and Ming Li.
\newblock Approximation algorithms for directed steiner problems.
\newblock \emph{Journal of Algorithms}, 33\penalty0 (1):\penalty0 73--91, oct
  1999, \doi{10.1006/jagm.1999.1042}.

\bibitem[Chitnis et~al.(2013)Chitnis, Hajiaghayi, and
  Kortsarz]{DBLP:conf/iwpec/ChitnisHK13}
Rajesh Chitnis, MohammadTaghi Hajiaghayi, and Guy Kortsarz.
\newblock Fixed-parameter and approximation algorithms: A new look.
\newblock In \emph{Parameterized and Exact Computation}, pages 110--122.
  Springer International Publishing, 2013, \doi{10.1007/978-3-319-03898-8\_11}.

\bibitem[Chitnis et~al.(2018)Chitnis, Feldmann, and
  Manurangsi]{chitnis2017parameterized}
Rajesh Chitnis, Andreas~Emil Feldmann, and Pasin Manurangsi.
\newblock {Parameterized Approximation Algorithms for Bidirected Steiner
  Network Problems}.
\newblock In \emph{26th Annual European Symposium on Algorithms (ESA 2018)},
  volume 112 of \emph{Leibniz International Proceedings in Informatics
  (LIPIcs)}, pages 20:1--20:16, Dagstuhl, Germany, 2018. Schloss
  Dagstuhl--Leibniz-Zentrum fuer Informatik.
\newblock ISBN 978-3-95977-081-1, \doi{10.4230/LIPIcs.ESA.2018.20}.

\bibitem[Chleb{\'{\i}}k and
  Chleb{\'{\i}}kov{\'{a}}(2002)]{chlebik2002approximation}
Miroslav Chleb{\'{\i}}k and Janka Chleb{\'{\i}}kov{\'{a}}.
\newblock Approximation hardness of the steiner tree problem on graphs.
\newblock In \emph{Algorithm Theory {\textemdash} {SWAT} 2002}, pages 170--179.
  Springer Berlin Heidelberg, 2002, \doi{10.1007/3-540-45471-3\_18}.

\bibitem[Cygan et~al.(2015)Cygan, Fomin, Kowalik, Lokshtanov, Marx, Pilipczuk,
  Pilipczuk, and Saurabh]{pc-book}
Marek Cygan, Fedor~V. Fomin, Lukasz Kowalik, Daniel Lokshtanov, D{\'{a}}niel
  Marx, Marcin Pilipczuk, Michal Pilipczuk, and Saket Saurabh.
\newblock \emph{{Parameterized Algorithms}}.
\newblock Springer, 2015, \doi{10.1007/978-3-319-21275-3}.

\bibitem[Dinur and Steurer(2014)]{dinur2014setcover}
Irit Dinur and David Steurer.
\newblock Analytical approach to parallel repetition.
\newblock In \emph{Proceedings of the 46th Annual {ACM} Symposium on Theory of
  Computing - {STOC} 14}. {ACM} Press, 2014, \doi{10.1145/2591796.2591884}.

\bibitem[Dom et~al.(2014)Dom, Lokshtanov, and Saurabh]{dom}
Michael Dom, Daniel Lokshtanov, and Saket Saurabh.
\newblock Kernelization lower bounds through colors and {IDs}.
\newblock \emph{{ACM} Transactions on Algorithms}, 11\penalty0 (2):\penalty0
  1--20, oct 2014, \doi{10.1145/2650261}.

\bibitem[Downey and Fellows(1999)]{downey2012parameterized}
R.~G. Downey and M.~R. Fellows.
\newblock \emph{Parameterized Complexity}.
\newblock Springer New York, 1999, \doi{10.1007/978-1-4612-0515-9}.

\bibitem[Dreyfus and Wagner(1971)]{DBLP:journals/networks/DreyfusW71}
Stuart~E. Dreyfus and Robert~A. Wagner.
\newblock The steiner problem in graphs.
\newblock \emph{Networks}, 1\penalty0 (3):\penalty0 195--207, 1971,
  \doi{10.1002/net.3230010302}.

\bibitem[Dvo\v{r}{\'{a}}k et~al.(2018)Dvo\v{r}{\'{a}}k, Feldmann, Knop,
  Masa\v{r}{\'{\i}}k, Toufar, and Vesel{\'{y}}]{STACS}
Pavel Dvo\v{r}{\'{a}}k, Andreas~Emil Feldmann, Du\v{s}an Knop, Tom{\'{a}}\v{s}
  Masa\v{r}{\'{\i}}k, Tom\'a\v{s} Toufar, and Pavel Vesel{\'{y}}.
\newblock Parameterized approximation schemes for steiner trees with small
  number of steiner vertices.
\newblock In \emph{35th Symposium on Theoretical Aspects of Computer Science,
  {STACS} 2018, February 28 to March 3, 2018, Caen, France}, pages 26:1--26:15,
  2018, \doi{10.4230/LIPIcs.STACS.2018.26}.

\bibitem[Eisenstat et~al.(2012)Eisenstat, Klein, and
  Mathieu]{eisenstat2012efficient}
David Eisenstat, Philip Klein, and Claire Mathieu.
\newblock An efficient polynomial-time approximation scheme for steiner forest
  in planar graphs.
\newblock In \emph{Proceedings of the Twenty-Third Annual {ACM}-{SIAM}
  Symposium on Discrete Algorithms}, pages 626--638. Society for Industrial and
  Applied Mathematics, jan 2012, \doi{10.1137/1.9781611973099.53}.

\bibitem[Feldmann and Marx(2016)]{DBLP:conf/icalp/FeldmannM16}
Andreas~Emil Feldmann and D{\'a}niel Marx.
\newblock {The Complexity Landscape of Fixed-Parameter Directed Steiner Network
  Problems}.
\newblock In \emph{43rd International Colloquium on Automata, Languages, and
  Programming (ICALP 2016)}, volume~55, pages 27:1--27:14, 2016,
  \doi{10.4230/LIPIcs.ICALP.2016.27}.

\bibitem[Fuchs et~al.(2007)Fuchs, Kern, Molle, Richter, Rossmanith, and
  Wang]{fuchs2007dynamic}
Bernhard Fuchs, Walter Kern, D~Molle, Stefan Richter, Peter Rossmanith, and
  Xinhui Wang.
\newblock Dynamic programming for minimum steiner trees.
\newblock \emph{Theory of Computing Systems}, 41\penalty0 (3):\penalty0
  493--500, 2007, \doi{10.1007/s00224-007-1324-4}.

\bibitem[Guo et~al.(2011)Guo, Niedermeier, and
  Such{\'{y}}]{DBLP:journals/siamdm/GuoNS11}
Jiong Guo, Rolf Niedermeier, and Ondrej Such{\'{y}}.
\newblock Parameterized complexity of arc-weighted directed {S}teiner problems.
\newblock \emph{{SIAM} J. Discrete Math.}, 25\penalty0 (2):\penalty0 583--599,
  2011, \doi{10.1137/100794560}.

\bibitem[Halperin and Krauthgamer(2003)]{approx-hardness}
Eran Halperin and Robert Krauthgamer.
\newblock Polylogarithmic inapproximability.
\newblock In \emph{Proceedings of the thirty-fifth {ACM} symposium on Theory of
  computing - {STOC} 03}, pages 585--594. {ACM} Press, 2003,
  \doi{10.1145/780542.780628}.

\bibitem[Hušek et~al.(2018)Hušek, Toufar, Knop, Masařík, and
  Eiben]{paceImplementation}
Radek Hušek, Tomáš Toufar, Dušan Knop, Tomáš Masařík, and Eduard Eiben.
\newblock Steiner tree heuristics for {PACE} 2018 challenge track {C}, 2018.
\newblock \url{https://github.com/goderik01/PACE2018}.

\bibitem[Hušek et~al.(2020)Hušek, Knop, and
  Masařík]{StarContractionHeuristic}
Radek Hušek, Dušan Knop, and Tomáš Masařík.
\newblock Approximation algorithms for steiner tree based on star contractions:
  A unified view, 2020, \eprint{2002.03583}.

\bibitem[Hwang et~al.(1992)Hwang, Richards, and Winter]{hwang1992steiner}
Frank~K Hwang, Dana~S Richards, and Pawel Winter.
\newblock \emph{The Steiner tree problem}, volume~53.
\newblock Elsevier, 1992, \doi{10.1016/s0167-5060(08)x7008-6}.

\bibitem[Jones et~al.(2013)Jones, Lokshtanov, Ramanujan, Saurabh, and
  Such{\'y}]{jones2013parameterized}
Mark Jones, Daniel Lokshtanov, MS~Ramanujan, Saket Saurabh, and Ond{\v{r}}ej
  Such{\'y}.
\newblock Parameterized complexity of directed steiner tree on sparse graphs.
\newblock In \emph{European Symposium on Algorithms (ESA)}, pages 671--682.
  Springer, 2013, \doi{10.1007/978-3-642-40450-4\_57}.

\bibitem[Karp(1972)]{MR51:14644}
Richard~M. Karp.
\newblock Reducibility among combinatorial problems.
\newblock In \emph{Complexity of computer computations}, pages 85--103. Plenum,
  1972, \doi{10.1007/978-1-4684-2001-2\_9}.

\bibitem[Lokshtanov et~al.(2017)Lokshtanov, Panolan, Ramanujan, and
  Saurabh]{lokshtanov2017lossy}
Daniel Lokshtanov, Fahad Panolan, M.~S. Ramanujan, and Saket Saurabh.
\newblock Lossy kernelization.
\newblock In \emph{Proceedings of the 49th Annual {ACM} {SIGACT} Symposium on
  Theory of Computing - {STOC} 2017}, pages 224--237. {ACM} Press, 2017,
  \doi{10.1145/3055399.3055456}.

\bibitem[Mölle et~al.(2007)Mölle, Richter, and
  Rossmanith]{molle2008enumerate}
Daniel Mölle, Stefan Richter, and Peter Rossmanith.
\newblock Enumerate and expand: Improved algorithms for connected vertex cover
  and tree cover.
\newblock \emph{Theory of Computing Systems}, 43\penalty0 (2):\penalty0
  234--253, oct 2007, \doi{10.1007/s00224-007-9089-3}.

\bibitem[Nederlof(2009)]{DBLP:conf/icalp/Nederlof09}
Jesper Nederlof.
\newblock Fast polynomial-space algorithms using m{\"{o}}bius inversion:
  Improving on steiner tree and related problems.
\newblock In \emph{Automata, Languages and Programming, 36th International
  Colloquium, {ICALP} 2009, Rhodes, Greece, July 5-12, 2009, Proceedings, Part
  {I}}, pages 713--725, 2009, \doi{10.1007/978-3-642-02927-1\_59}.

\bibitem[Pilipczuk et~al.(2014)Pilipczuk, Pilipczuk, Sankowski, and van
  Leeuwen]{planar-7}
Marcin Pilipczuk, Michal Pilipczuk, Piotr Sankowski, and Erik~Jan van Leeuwen.
\newblock Network sparsification for steiner problems on planar and
  bounded-genus graphs.
\newblock In \emph{2014 {IEEE} 55th Annual Symposium on Foundations of Computer
  Science}. {IEEE}, oct 2014, \doi{10.1109/focs.2014.37}.

\bibitem[Srikanta et~al.(2019)Srikanta, Laekhanukit, and Manurangsi]{dom-set}
Karthik~C. Srikanta, Bundit Laekhanukit, and Pasin Manurangsi.
\newblock On the parameterized complexity of approximating dominating set.
\newblock \emph{J. ACM}, 66\penalty0 (5), August 2019.
\newblock ISSN 0004-5411, \doi{10.1145/3325116}.

\bibitem[Such{\'y}(2017)]{suchy2017extending}
Ond{\v{r}}ej Such{\'y}.
\newblock Extending the kernel for planar steiner tree to the number of steiner
  vertices.
\newblock \emph{Algorithmica}, 79\penalty0 (1):\penalty0 189--210, 2017,
  \doi{10.1007/s00453-016-0249-1}.

\bibitem[Williamson and Shmoys(2011)]{williamson2011design}
David~P Williamson and David~B Shmoys.
\newblock \emph{The design of approximation algorithms}.
\newblock Cambridge university press, 2011, \doi{10.1017/cbo9780511921735}.

\bibitem[Zelikovsky(1993)]{zelikovsky_1993_11_over_6_apx}
Alexander Zelikovsky.
\newblock An $11/6$-approximation algorithm for the network {Steiner} problem.
\newblock \emph{Algorithmica}, 9:\penalty0 463--470, 1993,
  \doi{10.1007/BF01187035}.

\end{thebibliography}
\end{document}